\documentclass{amsart}
\usepackage{amsfonts}
\usepackage{amsmath}
\usepackage{amssymb}
\usepackage{amsthm}
\usepackage{color}
\usepackage{comment}
\usepackage{slashbox}

\usepackage{pdfpages}

\def\gathen#1{{#1}}

\usepackage[bookmarksopen=false,
breaklinks=true,%
    backref=page,pagebackref=true,plainpages=false,%
    hyperindex=true,pdfstartview=FitH,colorlinks=true,%
    pdfpagelabels=true,colorlinks=true,linkcolor=blue,citecolor=red]%
   {hyperref}

\newtheorem{Thm}{Theorem}
\newtheorem{Lem}[Thm]{Lemma}
\newtheorem{Prop}[Thm]{Proposition}
\newtheorem{Cor}[Thm]{Corollary}

\theoremstyle{definition}
\newtheorem{Def}[Thm]{Definition}

\theoremstyle{remark}
\newtheorem{Rem}[Thm]{Remark}

\numberwithin{equation}{section}

\def\QQ{{\mathbb{Q}}}
\def\DD{{\mathcal{D}}}

\def\NN{{\mathbb{N}}}
\def\PP{{\mathbb{P}}}
\def\KK{{\mathbb{K}}}

\def\PP{{\mathbb{P}}}

\def\LL{{\mathbb{L}}}

\def\bigOsoft{\tilde{\mathcal{O}}}

\def\bigO{\mathcal{O}}
\def\bigOsoft{\tilde{\mathcal{O}}}

\def\haty{\hat{y}}

\def\CD{{\mathcal D}} 

\title[Rational First Integrals of Polynomial Vector Fields]{Efficient Algorithms for Computing Rational First Integrals and Darboux Polynomials of Planar Polynomial Vector Fields}
\date{\today}

\author[A.~Bostan]{Alin Bostan}
\address{Alin Bostan: INRIA Saclay \^{I}le-de-France, B\^atiment Alan Turing\\
1 rue Honor\'e d'Estienne d'Orves\\
91120 Palaiseau, France}
\email{alin.bostan@inria.fr}

\author[G.~Ch\`eze]{Guillaume Ch\`eze}
\address{Guillaume Ch\`eze: Institut de Math\'ematiques de Toulouse\\
Universit\'e Paul Sabatier Toulouse 3 \\
MIP B\^at. 1R3\\
31 062 TOULOUSE cedex 9, France}
\email{guillaume.cheze@math.univ-toulouse.fr}

\author[T.~Cluzeau]{Thomas Cluzeau}
\address{Thomas Cluzeau: Universit\'e de Limoges ; CNRS ; XLIM UMR 7252 ; DMI \\
123 avenue Albert Thomas\\
87 060 LIMOGES cedex, France}
\email{cluzeau@ensil.unilim.fr}

\author[J.-A. Weil]{Jacques-Arthur Weil}
\address{Jacques-Arthur Weil: Universit\'e de Limoges ; CNRS ; XLIM UMR 7252 ; DMI \\
123 avenue Albert Thomas\\
87 060 LIMOGES cedex, France}
\email{jacques-arthur.weil@unilim.fr}

\date{\today}

\begin{document}
	
\begin{abstract}
We present fast algorithms for computing rational first integrals with bounded degree of a planar polynomial vector field. Our approach is inspired by an idea of Ferragut and Giacomini (\cite{FG}). We improve upon their work by proving that rational first integrals can be computed via systems of linear equations instead of systems of quadratic equations. The main ingredients of our algorithms are the calculation of a power series solution of a first order differential equation and the reconstruction of a bivariate polynomial annihilating a power series. This leads to a probabilistic algorithm with arithmetic complexity $\bigOsoft(N^{2 \omega})$ and to a deterministic algorithm solving the problem in $\bigOsoft(d^2N^{2 \omega+1})$ arithmetic operations, where~$N$ denotes the given bound for the degree of the rational first integral, and where $d \leq N$ is the degree of the vector field, and $\omega$  the exponent of linear algebra. We also provide a fast heuristic variant which computes a rational first integral, or fails, in $\bigOsoft(N^{\omega+2})$ arithmetic operations. By comparison, the best previous algorithm given in~\cite{Ch} uses at least $d^{\omega+1}\, N^{4\omega +4}$ arithmetic operations. We then show how to apply a similar method to the computation of Darboux polynomials. The algorithms are implemented in a Maple package {\sc RationalFirstIntegrals} which is available to interested readers with examples showing its efficiency.
\end{abstract}	
	
\maketitle
 
\section{Introduction}

\subsection*{Context} 
Let $\KK$ denote an effective field of characteristic zero, i.e, one can perform arithmetic operations and test equality of two elements (typically, $\KK=\QQ$ or $\QQ(\alpha)$, where $\alpha$ is an algebraic number).
Given two
polynomials $A,\, B$ in $\KK[x,y]$, we consider the planar polynomial vector
field
\begin{equation} \label{eq-sys-intro} \tag{$\mathsf{S}$} \quad \left\lbrace \begin{array}{rcl} \dot{x} &=& A(x,y), \\
\dot{y} &=& B(x,y), \end{array}\right.
\end{equation} 
and discuss the problem of computing
rational first integrals of \eqref{eq-sys-intro}, i.e., rational functions
$F \in \KK(x,y)\setminus \KK$ that are constant along the solutions $(x(t),y(t))$ of \eqref{eq-sys-intro}.
More precisely, the present article is concerned with the following algorithmic
problem:\\

 $({\mathcal P}_N)$: given a degree bound $N\in \NN$, either compute a rational first integral $F\in \KK(x,y)\setminus \KK$ of~\eqref{eq-sys-intro} of total degree
at most $N$, or prove that no such $F$ exists.\\

This old problem was already studied by Darboux in 1878 (\cite{Da}) and has
been the subject of numerous works ever since. The naive approach (by
indeterminate coefficients) leads to a polynomial system of quadratic
equations in the coefficients of $F$. Other methods use what is
called nowadays \emph{Darboux polynomials}, in the spirit of the celebrated
Prelle-Singer's method~\cite{PS}; see Subsection~\ref{review} for a review. These methods also require solving a polynomial system of quadratic equations.
Recently, Ch\`eze~\cite{Ch} has shown that problem $(\mathcal{P}_N)$ can be
solved in polynomial time in~$N$.  The importance of this result is mainly theoretical since the exponent in the polynomial complexity estimate is
bigger than~10.

To improve upon this current state of affairs, our starting point is the
article \cite{FG} of Ferragut and Giacomini. The key observation is that~\eqref{eq-sys-intro}
has a rational first integral if and only if all power series solutions in 
$\KK[[x]]$ of the first order non-linear differential equation 
	\begin{equation} \label{diff-eq-intro} \tag{$\mathsf{E}$} \quad \frac{dy}{dx}=\frac{B(x,y)}{A(x,y)}
	\end{equation}
are \emph{algebraic\/} over $\KK(x)$. Furthermore,
minimal polynomials of these algebraic power series lead to rational first
integrals.

The algorithm in~\cite{FG} still involves solving a polynomial system
of quadratic equations. Indeed, the  key observation above  is merely used to
reduce the number of equations in the quadratic system provided by the naive approach.

\subsection*{Our main contributions} In the present article, we push  further the observation of Ferragut and
Giacomini, so as to give fast algorithms solving Problem $({\mathcal
P}_N)$. In particular, we prove that this can be done by considering only
systems of {\em linear\/} equations instead of systems of quadratic equations.

We design a probabilistic algorithm that uses {$\bigOsoft(N^{2 \omega})$}
arithmetic operations in~$\KK$, where $\omega \in [2,3]$ is the exponent of
linear algebra over~$\KK$, and the soft-O notation $\bigOsoft(\,)$ indicates
that polylogarithmic factors are neglected. The probabilistic algorithm is
then turned into a deterministic one, that solves Problem $({\mathcal P}_N)$ in
arithmetic complexity $\bigOsoft(d^2\,N^{2 \omega+1})$, where $d=\max(\deg(A),\deg(B))$ denotes the degree of the polynomial vector field \eqref{eq-sys-intro}. This compares well to the
previous polynomial time algorithm given in~\cite{Ch}, which uses at least
$d^{\omega+1}N^{4\omega +4}$ arithmetic operations. Note that if we take $\omega=3$ (i.e., the cost of naive linear algebra),
then the above means that the best previously known complexity would be in $\bigOsoft(d^4N^{16})$ whereas our deterministic algorithm would use at most
$\bigOsoft(d^2N^7)$ arithmetic operations, and our probabilistic one would use $\bigOsoft(N^6)$. Lastly, we sketch a heuristic method that  uses
$\bigOsoft(N^{\omega+2})$ arithmetic operations (i.e., $\bigOsoft(N^{5})$ using classical linear algebra) which is sub-cubic, given that the output has size $\bigO(N^2)$.

We provide algorithmic details, notably precise degree bounds and
complexity estimates. The algorithms developed in the article are implemented in a Maple package called {\sc RationalFirstIntegrals} which is available with various examples at \url{http://www.ensil.unilim.fr/~cluzeau/RationalFirstIntegrals.html}. Using this implementation, we demonstrate the
efficiency of our algorithms on some examples. Finally, we show how to apply a similar method to the computation of Darboux polynomials.

\subsection*{Structure of the article}
In Section~\ref{sec:review}, we recall Darboux's approach to the integrability of
polynomial vector fields, related works, and existing results about the problem $({\mathcal P}_N)$. We also give useful facts on the so-called {\em
spectrum problem}. We recall in Section~\ref{sec:RFI+deq}
the connection between rational first integrals of the polynomial vector field
\eqref{eq-sys-intro} and algebraic power series solutions of $\eqref{diff-eq-intro}$. We then propose a first
 algorithm, based on linear algebra, that solves
Problem~$({\mathcal P}_N)$. Building on this, we develop in Section~\ref{section-proba_det_algo} an
efficient probabilistic algorithm, and then turn it into an efficient deterministic algorithm. In Section~\ref{complexity-section}, we study the
arithmetic complexity of the algorithms developed in Section~\ref{section-proba_det_algo}, and discuss several
algorithmic issues. Then, in Section~\ref{sec_impl_exp} we present our implementation and
display its behavior on various examples. Finally, Section~\ref{sec:Darboux} shows how similar
ideas can be used for computing the set of all irreducible Darboux polynomials (of a given degree)
of planar polynomial vector fields.

\subsection*{Notation}

The degree $\deg(P)$ of a bivariate polynomial $P \in \KK[x,y]$ is the total
degree of $P$. A rational function $P/Q$ with $P,Q \in \KK[x,y]$ is said to be
{\em reduced} when $P$ and $Q$ are coprime. The degree $\deg(F)$ of a reduced
rational function $F=P/Q$ is the maximum of $\deg(P)$ and $\deg(Q)$.  \\
We denote
by $\overline{\KK}$ an algebraic closure of the field~$\KK$.\\
 We write
$\dot{f}:=\frac{\partial f}{\partial t}$ for the usual formal derivative of
the ``function" (polynomial, or power series)~$f$ with respect to the
variable $t$. \\
For a set $\Omega$, we denote by $| \Omega |$ its cardinality.

\section{Review on first integrals of polynomial vector fields}\label{sec:review}

In this section, we recall several useful facts, mainly to keep the exposition as self-contained as possible, and to clarify
the understanding of the algorithms that we develop below. Some results are not original.

\subsection{First definitions and classical results}

We consider an autonomous planar polynomial vector field 
\begin{equation} \label{eq-sys} \tag{$\mathsf{S}$} \quad \left\lbrace \begin{array}{rcl} \dot{x} &=& A(x,y), \\
\dot{y} &=& B(x,y), \end{array}\right.
\end{equation}
where $x$ and $y$ are unknown ``functions" of the time variable~$t$, $A$ 
and $B$ are polynomials in $\KK[x,y]$, and $d:=\max(\deg(A),\deg(B))$ denotes the degree of the polynomial vector field. 
Without any loss of generality, $A$ and $B$ will be assumed to be coprime in the remaining of the article.

To \eqref{eq-sys} is attached the \emph{derivation}
$$\DD:=A(x,y)\,\frac{\partial}{\partial x} + B(x,y)\,\frac{\partial}{\partial
y},$$
acting on the polynomial ring $\KK[x,y]$. We thus view $\KK[x,y]$
as a differential ring endowed with the derivation $\DD$. We denote by
$\KK(x,y)$ its field of fractions. 

\begin{Def} \label{def:rfi}
A {\em rational first integral of \eqref{eq-sys}} 
is a non-constant rational function $F \in \KK(x,y)\setminus \KK$ satisfying $\DD(F)=0$. 
\end{Def}

A rational first integral $F$ of \eqref{eq-sys} is thus a non-trivial
constant for the derivation~$\DD$. Intuitively, this means that if
$(x(t),y(t))$ is a pair of ``functions" satisfying~\eqref{eq-sys}, then $F(x(t),y(t))$
is constant when $t$ varies. We explain in Theorem~\ref{thm11} below why no algebraic extension of the base field is necessary in Definition~\ref{def:rfi}.

\medskip A starting observation is that the rational function $F=P/Q$ is a
first integral for~\eqref{eq-sys} if and only if $\DD(P) \, Q = \DD(Q) \, P$. Therefore, if
$F$ is a \emph{reduced\/} rational first integral for~\eqref{eq-sys}, then $P$ divides
$\DD(P)$, and $Q$ divides $\DD(Q)$ in $\KK[x,y]$. This motivates the
following definition.

\begin{Def} \label{def:Darboux}
A polynomial $M \in \overline\KK[x,y] \setminus \overline\KK$ is a {\em Darboux polynomial for $\DD$} if $M$ divides $\DD(M)$ in $\overline\KK[x,y]$. Therefore, if $M$ is a Darboux polynomial for $\DD$, then there exists a polynomial $\Lambda \in \overline\KK[x,y]$ such that $\DD(M)=\Lambda\,M$. Such a polynomial $\Lambda \in \overline\KK[x,y]$ is called a {\em cofactor associated with the Darboux polynomial $M$}.
\end{Def}

Darboux polynomials were introduced by G. Darboux in \cite{Da}. These polynomials correspond to algebraic curves invariant under the vector field. The following lemma will be used in the sequel: it means that if we have a non-singular initial condition, then there is a unique irreducible invariant algebraic curve satisfying this initial condition, see \cite[Lemma~A.1]{Si}.
\begin{Lem}\label{lem:uniq}
Let $\DD=A(x,y)\,\frac{\partial}{\partial x} + B(x,y)\,\frac{\partial}{\partial
y}$ be the derivation attached to~\eqref{eq-sys} and let $(x_0, y_0)$ be a non-singular point of $\DD$, i.e., $A(x_0,y_0) \neq 0$ or $B(x_0,y_0) \neq 0$. If $M_1$
and $M_2$ are two Darboux polynomials for $\DD$ such that $M_1(x_0,y_0) =M_2(x_0,y_0)= 0$ and if $M_1$ is irreducible, then $M_1$ divides $M_2$.
\end{Lem}

Darboux polynomials are sometimes called \emph{partial first integrals\/} in the
literature. The reason is that rational first integrals and Darboux
polynomials are intimately related notions: as sketched above, numerators and
denominators of reduced rational first integrals are Darboux polynomials. The
converse is also true, see Corollary~\ref{RFI-Darboux} below.

A fundamental property of Darboux polynomials is given in the following lemma
(see, e.g., \cite[Lemma 8.3, p. 216]{DLA}) that can be proved by a
straightforward calculation.

\begin{Lem}\label{darboux-semi-group}
Let $M \in \overline\KK[x,y]$ and let $M=M_1\,M_2$ be a factorization of $M$ in $\overline\KK[x,y]$, with $M_1$ and $M_2$ coprime.
Then, $M$ is a Darboux polynomial for $\DD$ if and only if $M_1$ and $M_2$ are Darboux polynomials for $\DD$. Furthermore, if $\Lambda_M, \, \Lambda_{M_1}$ and $\Lambda_{M_2}$ denote respectively the cofactors of $M, \, M_1$ and $M_2$, then $\Lambda_M=\Lambda_{M_1}+\Lambda_{M_2}$.
\end{Lem}

As a corollary we get:

\begin{Cor} \label{RFI-Darboux} 
Let $F = P/Q$ be a reduced rational function in $\KK(x,y)$. Then $F$ is a rational first integral of \eqref{eq-sys} if and only if $P$ and $Q$ are Darboux polynomials for~$\DD$ with the same cofactor.
\end{Cor}

 The previous corollary gives a relation between Darboux polynomials and rational  first integrals. The next theorem shows that if we have enough Darboux polynomials, then we have a rational first integral, see  \cite[Appendix]{Si} for a modern proof.

\begin{Thm} \label{thm:darboux}[Darboux-Jouanolou~\cite{Da,Jou,Si}]\\
If $d=\max(\deg(A),\deg(B))$, then the polynomial vector field \eqref{eq-sys} has a reduced rational first integral $P/Q$ if and only if $\DD$ has at least $d\,(d+1)/2+2$ irreducible Darboux polynomials. In this case, $\DD$ has infinitely many irreducible Darboux polynomials and any of them divides a linear combination $\lambda\,  P - \mu \, Q$, for some $\lambda,\, \mu \in \overline{\KK}$ not both zero. 
Moreover, all but finitely many irreducible Darboux polynomials  are of the form $\lambda \,  P - \mu \, Q$ and have the same degree.
\end{Thm}

A useful corollary of Theorem~\ref{thm:darboux} is the following, see \cite{Si}:

\begin{Cor} \label{degre-fini}
For each planar polynomial vector field \eqref{eq-sys}, there exists a non-negative integer $N_{\eqref{eq-sys}}$ such that any {\em irreducible} Darboux
polynomial for the derivation $\DD$ attached to~\eqref{eq-sys} has degree at most $N_{\eqref{eq-sys}}$.
\end{Cor}

Given a derivation $\DD$, the problem of finding a bound for the degree of irreducible
Darboux polynomials is known to be difficult: this is the so-called {\em
Poincar\'e problem}. It has been deeply studied in the literature and many
partial results exist (\cite{Poi,CeLi91,Car,Per02,Walcher,ChaGar,LeiYan} and others) though the
question is not fully solved yet.
The fact that the derivation $\DD=n\,x\,\frac{\partial}{\partial x} + y\,\frac{\partial}{\partial y}$
with $n \in \NN^*$ admits $x-y^{n}$ as an irreducible Darboux polynomial
shows that a bound depending only on the degrees of the entries cannot exist:
arithmetic conditions on the coefficients of $\DD$ have to be taken into
account as well.\\

Consequently, given a planar polynomial vector field \eqref{eq-sys}, or equivalently a derivation $\DD$,  two distinct problems occur when we want to compute rational first integrals: 
\begin{enumerate}
\item Find a bound on the degree of the numerator and denominator of a rational first integral, that is a bound on the degree of irreducible Darboux polynomials; 
\item $({\mathcal P}_N)$: given a degree bound $N\in \NN$, either compute a rational first integral $F\in \KK(x,y)\setminus \KK$ of \eqref{eq-sys} of total degree
at most $N$, or prove that no such $F$ exists. 
\end{enumerate}

Our aim is to give an efficient algorithm to handle the second problem $({\mathcal P}_N)$.\\

In this article we suppose that $d \leq N$. This hypothesis is natural because if a derivation has a polynomial first integral of degree $N$, then we can show that $d \leq N-1$, see \cite[Theorem 6]{FerLib} or \cite{Poi}.

\subsection{Non-composite rational functions and their spectrum} \label{sec_spectrum}  

We recall here the definition of \emph{composite} rational functions and what is called the \emph{spectrum} of a rational function. We then use these notions to describe the kernel of the derivation $\CD$ and to give some of its properties.

\begin{Def} 
	A rational function $F(x,y) \in \KK(x,y)$ is \emph{composite} if
it can be written $F=u \circ h$, i.e., $F=u(h)$, where $h \in \KK(x,y)\setminus \KK$ and $u
\in \KK(T)$ with $\deg(u) \geq 2$. Otherwise $F$ is said to be {\em
non-composite}. 
\end{Def}

 In \cite{Ollagnier,CheDecompDarboux}, the authors propose different algorithms for the decomposition of rational functions using properties of Darboux polynomials and rational first integrals of the Jacobian derivation.

\begin{Lem}  \label{lem:RFI-alg}
The set of all rational first integrals of \eqref{eq-sys} is a $\KK$-algebra. It is closed under composition with rational functions in $\KK(T)$,
and moreover, $F$ is a rational first integral of \eqref{eq-sys} if and only if $u \circ F$ is a rational first integral of \eqref{eq-sys} for some $u \in \KK(T) \setminus \KK.$
\end{Lem}

\begin{proof}
The first assertion directly follows from the fact that the derivation\\
$\DD=A(x,y)\,\frac{\partial}{\partial x} + B(x,y)\,\frac{\partial}{\partial
y}$
is $\KK$-linear and satisfies Leibniz's rule 
\[ \DD(F_1 F_2) = F_1 \, \DD(F_2) + \DD(F_1) \, F_2.\]
The second assertion follows from the equality 
\[ \DD(u \circ F) = u'(F) \, \DD(F),\]
and the fact that $u'(F)$ is zero if and only if $u \in \KK$.
\end{proof}

A more precise version of Lemma~\ref{lem:RFI-alg} is given by the next theorem which completely describes the $\KK$-algebra structure of the set of all rational first integrals of \eqref{eq-sys}. This theorem seems to be a folklore result but we have not found a suitable reference. Consequently, a complete proof is provided here.

\begin{Thm}  \label{thm:RFI-struct}
Let $\DD$ be the derivation attached with \eqref{eq-sys}. Then we have:
\[ \{ G \in \KK(x,y) \; | \; \DD(G) = 0 \} = \KK(F),\]
for some non-composite reduced rational first integral $F$ of \eqref{eq-sys}.\\
Then any other rational first integral $G$ of \eqref{eq-sys} is of the form 
$G = u \circ F$ for some $u \in \KK(T) \setminus \KK.$
In particular, any two non-composite reduced 
rational first integrals are equal, up to a homography.
\end{Thm}

\begin{proof}
Let  $\LL=\{ G \in \KK(x,y) \; | \; \DD(G) = 0 \} $. We have $\KK \subset \LL \subset \KK(x,y)$, so, from \cite[Theorem 1, p. 12]{Schinzel}, we deduce that $\LL$ is finitely generated over $\KK$ and that  $\LL=\KK(f_1, f_2, f_3)$ for some $f_1,\,f_2,\,f_3 \in \KK(x,y)$.\\
As for $i \in \{1,2,3\}$,  $A(x,y)\,\frac{\partial f_i}{\partial x}+B(x,y)\,\frac{\partial f_i}{\partial y}=0$, we get that:
$$\frac{\partial f_i}{\partial y}\frac{\partial f_j}{\partial x}-\frac{\partial f_i}{\partial x}\frac{\partial f_j}{\partial y}=0, \quad  \text{for all} \; i, \, j \in \{1,2,3\}.$$
The Jacobian criterion implies that $f_1, f_2, f_3$ are algebraically dependent and thus the transcendence degree of $\LL$ over $\KK$ is equal to one. By the extended Luroth's theorem, see \cite[Theorem 3, p. 15]{Schinzel}, we get $\LL=\KK(F)$, for $F \in \KK(x,y)$. In particular, $F$ is a rational first integral of \eqref{eq-sys}.\\
Now, if $F$ is composite, $F=u(H)$, with $\deg(u)\geq 2$, then \mbox{$\KK(F) \subsetneq \KK(H)$}, see, e.g., \cite[Proposition 2.2]{GutRubSev}. By Lemma \ref{lem:RFI-alg}, $H$ is also a rational first integral of \eqref{eq-sys} so that $H \in \LL$ and $\KK(H) \subset \LL$. This yields $\LL =\KK(F) \subsetneq \KK(H) \subset \LL$, which is absurd. Thus $F$ is non-composite which gives the desired result.
 \end{proof}

As a consequence of Theorem~\ref{thm:RFI-struct}, \emph{non-composite reduced\/}
rational first integrals coincide with rational first integrals with minimal
degree; they will play a key role in the remaining of this text.\\

In Definition~\ref{def:rfi}, we have defined rational first integrals as elements of $\KK(x,y)$. The next theorem explains why it is in general not necessary to consider rational first integrals in $\overline{\KK}(x,y)$. To our knowledge this result is proved here for the first time. In \cite{Man}, the authors show that if there exists a rational first integral in $\overline{\KK}(x,y)$, then there also exists a rational first integral in $\KK(x,y)$. We improve this result by taking into account the degrees of these rational first integrals.
\begin{Thm}\label{thm11}
If \eqref{eq-sys} admits a non-composite rational first integral in $\overline{\KK}(x,y)$, then it admits a non-composite rational first integral in $\KK(x,y)$ with the same degree.
\end{Thm}

\begin{proof}
Let $f \in \overline{\KK}(x,y)$ be a non-composite rational first integral of \eqref{eq-sys}. We denote by $N(f)$ the product
$$N(f)=\prod_{\sigma_i \in G} \sigma_i(f),$$
where $G$ is the Galois group over $\KK$ of the smallest Galois extension containing all the coefficients of $f$, and we have $N(f) \in \KK(x,y)$.
As $A, B \in \KK(x,y)$, $N(f)$ is also a rational first integral of \eqref{eq-sys}. Thus, by Lemma~\ref{lem:RFI-alg}, there exists a non-composite rational first integral $F \in \KK(x,y)$ of \eqref{eq-sys}. Now, applying Theorem~\ref{thm:RFI-struct} with ground field $\overline{\KK}$ instead of $\KK$, we get that $F=u(f)$, with $u \in \overline{\KK}(T)$. Furthermore, by \cite[Theorem 13]{BCN}, $F$ is non-composite in $\KK(x,y)$ implies that $F$ is non-composite in $\overline{\KK}(x,y)$. It thus follows that $\deg(u)=1$ so that $\deg(F)=\deg(f)$.
\end{proof}


Now, we introduce the \emph{spectrum} of a rational function which will play a crucial role in our algorithms.

\begin{Def}
Let $P/Q  \in \KK(x,y)$ be a reduced rational function of degree $N$. The set 
\begin{eqnarray*}
\sigma(P,Q) =\{ (\lambda:\mu) \in \PP^1_{\overline{\KK}} &\mid &\lambda \, P- \mu \, Q\textrm{ is reducible in } \overline{\KK}[x,y], \\
&&\textrm{ or } \deg (\lambda \, P - \mu \, Q ) < N\}
\end{eqnarray*} 
is called the \emph{spectrum} of $P/Q$.
\end{Def}

In the context of rational first integrals of polynomial vector fields, the elements of the spectrum are sometimes called \emph{remarkable values}, see, e.g., \cite{FerLib}. There exists a vast bibliography about the spectrum,  see for example \cite{Rup,Lor,Vis,Vis1,AHS,Bod, BC}.\\

The spectrum $\sigma(P,Q)$ is finite if and only if $P/Q$ is non-composite and
if and only if the pencil of algebraic curves $\lambda \, P - \mu \, Q=0$ has an
irreducible general element, see for instance \cite[Chapitre 2, Th\'eor\`eme
3.4.6]{Jou} or \cite[Theorem 2.2]{Bod} for detailed proofs.

To the authors' knowledge, the first effective result on the spectrum is due
to Poincar\'e. In~\cite{Poi}, he establishes a relation between the number of
\emph{saddle points} and the cardinal of the
associated spectrum, in the case where all the singular points of the polynomial vector field
are distinct. In particular this yields the bound $ |\sigma(P,Q)| \leq d^2$ on the cardinality of the spectrum.
This bound was improved recently in \cite{Che12}:

\begin{Thm}\label{bornespectre} 
Let $\DD$ be the derivation attached with \eqref{eq-sys} and $d$ denotes the degree of \eqref{eq-sys}.
If $P/Q$ is a reduced non-composite rational  first integral of  \eqref{eq-sys}, then:
$$|\sigma(P,Q)| \leq \mathcal{B}(d)+1, \textrm{ where } \mathcal{B}(d)=\dfrac{d(d+1)}{2}.$$
\end{Thm}

As a consequence of Theorem~\ref{bornespectre} and
Corollary~\ref{RFI-Darboux}, if $P/Q$ is a reduced non-composite rational first
integral of \eqref{eq-sys}, then for all but $\mathcal{B}(d)+1$ constants $\sigma
\in \overline{\KK}$, the polynomial $P-\sigma\,Q$ has degree~$N$ and is an
irreducible Darboux polynomial for $\DD$ with the same cofactor as $P$
and~$Q$. This means that if \eqref{eq-sys} has a rational first integral, there exist an infinite number of irreducible Darboux polynomials which all have the same
degree (and the same cofactor).\\

The following lemma will be useful in Section~\ref{section-proba_det_algo} for the study of our probabilistic and deterministic algorithms. 
\begin{Lem} \label{lem:remarkable_values}
If $P/Q \in \KK(x,y)$ is a reduced non-composite rational function of degree at most~$N$, then the number of values of $c\in\overline\KK$ for which $(Q(0,c): P(0,c))$ belongs to $\sigma(P,Q)$ is bounded by $N \, (\mathcal{B}(d)+1)$.
\end{Lem}

\begin{proof} 
Let $c\in\overline\KK$ be such that $(Q(0,c ):P(0,c)) \in \sigma(P,Q)$. By Theorem~\ref{bornespectre}, $\sigma(P,Q)$ contains at most $\mathcal{B}(d)+1$ elements. Now, for each $(\lambda:\mu) \in \sigma(P,Q)$, as $P$ and $Q$ are of degree at most $N$, there exist at most $N$ values of $c$ such that $\lambda \, Q(0,c)-\mu \, P(0,c)=0$. This ends the proof.
\end{proof}

\subsection{Existing algorithms for computing rational first integrals  and Darboux polynomials of bounded degree} \label{review}
There is a vast literature regarding the computation of rational or elementary first integrals (see for example \cite{FG,Man,Ch,PS,ScGuRa90,Si,LibZha,DuaDuaMotSke,Poi}) and Darboux polynomials (see for example \cite{CouSch09,CouSch06,ChaGiaGra,Wei95,Da}). Note that, among these articles, very few restrict to the specific question of rational first integrals. Surveys on computing first integrals (not restricted to planar systems) can be found for example in \cite{Go01,Sc93,DLA} and \cite{Ol93} for symmetry methods which we do not address here.

Given a degree bound~$N$, the \emph{naive\/} approach to solve $({\mathcal
P}_N)$ consists in using the method of undetermined coefficients. This leads
to a system of polynomial (quadratic) equations in the unknown coefficients of
the rational first integral, see \cite{Ch} for a complexity estimate of this
approach.\\

Interest in Darboux polynomials has been revived by the appearance of the Prelle-Singer's method, \cite{PS, Man,DuaDuaMotSke}.
In \cite{CouSch06, CouSch09}, Coutinho and Menasch\'e Schechter give  necessary conditions for the existence of Darboux polynomials.  Other necessary conditions are contained in \cite{ChaGiaGinLib, ChaGiaGra} and also in works on inverse integrating factors \cite{ChaGiaGinLib,CFL}. The bottleneck of the Prelle-Singer's method and all of its variants is the computation \emph{by undetermined coefficients\/} of all irreducible Darboux polynomials of bounded degree, which leads again to solving a polynomial system. This yields an exponential complexity algorithm, see \cite{Ch}. \\

In \cite{Ch}, Ch\`eze shows that if the
derivation~$\DD$ 
admits only finitely many irreducible
Darboux polynomials of degree at most $N$, then it is possible to compute all of them 
by using the so-called \emph{ecstatic curve} introduced in \cite{Per01} within a number of \emph{binary
operations} that is \emph{polynomial\/} in the bound~$N$,  in the degree
$d=\max(\deg(A),\deg(B))$ of $\DD$ and in the logarithm of the height of $A$ and $B$.  A nontrivial modification of this algorithm provides a polynomial-time method to solve $({\mathcal P}_N)$, see again \cite{Ch}. To our knowledge, this is the first algorithm
solving~$({\mathcal P}_N)$ in polynomial-time. Unfortunately, the exponent is
quite large, making the algorithm unpractical even for moderate values of~$N$.
This drawback is due to the fact that algorithm \textsf{Rat-First-Int}
in~\cite{Ch} needs to compute the irreducible factors of a bivariate
polynomial of degree $\approx dN^4$, and the best known algorithms for solving
this subtask have arithmetic complexity, e.g., $\bigOsoft(d^{\omega +1} N^{4\omega
+4})$, see \cite{BLSSW,Lec}. \\

Last, we mention the article \cite{FG} of Ferragut and Giacomini, where the algebraicity of a generic power
series solution of the differential equation $\eqref{diff-eq-intro}$ is used to improve the
efficiency of the naive algorithm. Precisely, the system of quadratic
equations yielding the coefficients of a rational first integral is reduced to
a simpler system of (still quadratic) equations. Although this gives a good
heuristic improvement on the naive method, we show in the present article how to
turn it into a fast algorithm. Starting from the link between \eqref{eq-sys}
and $\eqref{diff-eq-intro}$, we reduce $({\mathcal P}_N)$ to solving a system of
\emph{linear} equations. Furthermore, we give tight bound on the number
of terms of power series solutions of $\eqref{diff-eq-intro}$ that are sufficient to detect the
existence of rational first integrals and to compute one of them when it exists.  
This enables us to turn
the heuristic in~\cite{FG} into an algorithm with polynomial complexity, that
is more efficient both in theory and in practice than all previous algorithms, see Section~\ref{sec_impl_exp}.

\section{Rational first integrals, differential equations and algebraic power series} \label{sec:RFI+deq}
\subsection{Algebraic power series and rational first integrals}

\begin{Def}
A formal power series $y(x) \in \KK[[x]]$ is said to be {\em algebraic} if it is algebraic over $\KK(x)$, that is, if there exists a non-zero polynomial $M \in \KK[x,y]$ such that 
$M(x,y(x))=0$. An irreducible polynomial $M \in \KK[x,y]$  satisfying $M(x,y(x))=0$ 
is called a {\em minimal polynomial of $y(x)$} in $\KK[x,y]$.
 \end{Def}
 
With the planar polynomial vector field \eqref{eq-sys}, we associate the first order non-linear differential equation: 
\begin{equation} \label{diff-eq} \tag{$\mathsf{E}$} \quad \frac{dy}{dx}=\frac{B(x,y)}{A(x,y)}
\end{equation}

We may assume without any loss
of generality that $x$ does not divide $A$, i.e., $A(0,y) \not \equiv 0$. We will explain how we can reduce to 
this situation and study the complexity of this reduction in Subsection~\ref{subsec:regular}.\\

Then, the formal version of the Cauchy-Lipschitz theorem for non-linear
(first-order) differential equations ensures that for any $c \in \KK$ such
that $A(0,c)\neq 0$, the equation \eqref{diff-eq} admits a unique power series solution
$y_{c}(x) \in \KK[[x]]$ satisfying $y_{c}(0)=c$. Note that high-order truncations of the power series $y_{c}(x)$ can be computed efficiently
using the algorithm of Brent and Kung~ \cite{BK}.\\

The following standard result is fundamental to both our method and the one of Ferragut and Giacomini in \cite{FG}.

\begin{Prop} \label{lem_alg}  
Consider the planar polynomial vector field \eqref{eq-sys} and assume that $A(0,y)\not \equiv 0$. Let $c \in \KK$ satisfy $A(0,c)\neq 0$ and $y_{c}(x)\in \KK[[x]]$ be the unique power series solution of \eqref{diff-eq} such that $y_{c}(0)=c$. 

\begin{enumerate}
\item If \eqref{eq-sys} admits a non-composite rational first integral $P/Q$, then the power series $y_{c}(x)$ is algebraic. More precisely, $y_c(x)$ is a root of the
non-zero polynomial $\lambda \, P-\mu \, Q$, where $\lambda=Q(0,c)$ and $\mu=P(0,c)$.

\item If $P/Q$ is a reduced non-composite rational first integral of \eqref{eq-sys} of degree at most $N$, then, for all but at most $N\,(\mathcal{B}(d)+1)$ values of $c\in \KK$, the polynomial $\lambda \, P-\mu \, Q$, where $\lambda=Q(0,c)$ and $\mu=P(0,c)$ is a minimal polynomial of $y_{c}(x)$.  
\end{enumerate}

\end{Prop}

\begin{proof}
Let $F = P/Q$ be a non-composite rational first integral of \eqref{eq-sys}. Since  the spectrum $\sigma(P,Q)$ is finite,  we can suppose that $P$ and $Q$ are irreducible and coprime. Let us first show that $P(0,c)\neq 0$ or $Q(0,c) \neq 0$. As $(0,c)$ is a non-singular point of $\DD$, if $P(0,c)=Q(0,c)=0$, then Lemma \ref{lem:uniq} implies that $P=\alpha\,Q$ with $\alpha \in \KK$. As $P$ and $Q$ are coprime, we get a contradiction so that necessarily $P(0,c)\neq 0$ or $Q(0,c) \neq 0$.\\ 
We thus suppose $Q(0,c)\neq 0$ else we consider the rational first integral $Q/P$ instead of $P/Q$. As $Q(0,c)\neq 0$, the power series $Q(x,y_{c}(x))$ is invertible so that $\DD(F)=0$ yields $\DD(F)(x,y_{c}(x))=0$. The latter equality can be written 
\begin{eqnarray*}
A(x,y_{c}(x)) \, \dfrac{\partial F}{\partial x}(x,y_{c}(x))+B(x,y_{c}(x))\, \dfrac{\partial F}{\partial y} (x,y_{c}(x)) & = & 0.
\end{eqnarray*}
Dividing the equality by the invertible power series $A(x,y_{c}(x))$ and using the fact that $y_{c}(x)$ is a solution of \eqref{diff-eq}, we obtain
\begin{eqnarray*}
\dfrac{\partial F}{\partial x}(x,y_{c}(x))+\dfrac{B(x,y_{c}(x))}{A(x,y_{c}(x))} \, \dfrac{\partial F}{\partial y} (x,y_{c}(x)) & = &0,\\
\dfrac{\partial F}{\partial x}(x,y_{c}(x))+\dfrac{d y_{c}(x)}{dx} \, \dfrac{\partial F}{\partial y} (x,y_{c}(x)) & = & 0, \\
\dfrac{d \big(F (x,y_{c}(x))\big)}{dx} & = & 0.
\end{eqnarray*}
It follows that $F(x,y_{c}(x))=\sigma_c$, for some $\sigma_c \in \KK$. Consequently, we have 
$$ P(x,y_{c}(x)) - \sigma_c \, Q(x,y_{c}(x))=0,$$ with $\sigma_c =P(0,c)/Q(0,c)$ which proves~(1). \\
If, in addition, $F=P/Q$ is reduced non-composite of degree at most $N$, then (2) follows directly from Lemma~\ref{lem:remarkable_values}.
\end{proof}

Proposition~\ref{lem_alg} shows in particular that if \eqref{eq-sys} has a rational first
integral $P/Q$, then all power series solutions of \eqref{diff-eq} are algebraic. The next proposition which is well known (see \cite{FG,Wei95}) asserts that the converse is also true.

\begin{Prop} \label{RFI-Darboux-Diffeq}
Let \eqref{eq-sys} be a planar polynomial vector field, $\DD$ the associated derivation, and \eqref{diff-eq} be the associated differential equation.
\begin{enumerate}
  \item If $M \in \KK[x,y]$ is an irreducible Darboux polynomial for ~$\DD$, then all roots $y(x)\in \overline{\KK(x)}$ of $M$ such that $A(0,y(0)) \neq 0$ are power series solutions of~\eqref{diff-eq}.
  \item The minimal polynomial of an algebraic solution $y(x)\in \KK[[x]]$ of \eqref{diff-eq} such that $A(0,y(0)) \neq 0$ is a Darboux polynomial for $\DD$.
  \item \eqref{eq-sys} admits a rational first integral if and only if all the power series solutions of \eqref{diff-eq} are algebraic.
\end{enumerate}
\end{Prop}

\begin{proof}
  Assume first that $M \in \KK[x,y]$ is an irreducible Darboux polynomial for $\DD$, and that $y(x) \in \overline{\KK(x)}$ is a root of $M$, i.e., $M(x,y(x))=0$. Since $M$ divides $\DD(M)$, it follows that $y(x)$ is also a root of $\DD(M)$. This implies 
\[\dfrac{\partial M}{\partial x}(x,y(x))= - \dfrac{B(x,y(x))}{A(x,y(x))}\,\dfrac{\partial M}{\partial y} (x,y(x)).\]
On the other hand, $M(x,y(x))=0$ implies by differentiation with respect to $x$ that 
\[\dfrac{\partial M}{\partial x}(x,y(x))= - \dfrac{d y(x)}{dx}\,\dfrac{\partial M}{\partial y} (x,y(x)).\]
These two equalities provides 
\[\dfrac{\partial M}{\partial y} (x,y(x)) \,  \left( \dfrac{d y(x)}{dx} - \dfrac{B(x,y(x))}{A(x,y(x))} \right) = 0.\]
As $M$ is irreducible, $\dfrac{\partial M}{\partial y} (x,y(x)) \neq 0$ so that $y(x)$ is a solution of \eqref{diff-eq}. As mentioned before $A(0,y(0)) \neq 0$ implies $y(x) \in \KK[[x]]$ (Cauchy-Lipschitz theorem) which proves~(1).

Assume now that $y(x)\in \KK[[x]]$ is an algebraic solution of \eqref{diff-eq} such that \\
$A(0,y(0)) \neq 0$, and let $M$ be its minimal polynomial. Then
\[ 0 = \dfrac{d \big(M (x,y(x))\big)}{dx}
=\dfrac{\partial M}{\partial x}(x,y(x))+\dfrac{dy(x)}{dx}\,\dfrac{\partial M}{\partial y} (x,y(x)),
\] 
and since $y(x)$ is a solution of \eqref{diff-eq}, the latter equality implies that $y(x)$ is also  a root of $\DD(M)$. Now as $M$ is the minimal polynomial of $y(x)$, it follows that $M$ divides $\DD(M)$, i.e., $M$ is a Darboux polynomial for $\DD$ and we have proved~(2).

Let us now prove (3). If all the power series solutions of \eqref{diff-eq} are algebraic, then by (2), $\DD$ admits infinitely many Darboux polynomials. Then  Theorem~\ref{thm:darboux} shows that \eqref{eq-sys} admits a rational first integral. The proof ends here since the other implication of (3) has been proved in Proposition~\ref{lem_alg}. 
\end{proof}

 \subsection{Algebraic power series solutions of \eqref{diff-eq} }

We have seen in Proposition~\ref{lem_alg} that if $P/Q$ is a reduced non-composite rational first integral of \eqref{eq-sys}, then  a minimal polynomial of a power series solution of \eqref{diff-eq} is generically of the form $\lambda \, P- \mu \, Q$ for some constants $\lambda$ and $\mu$. Thus, if we are able to compute such a minimal polynomial, we can deduce the rational first integral $P/Q$. In practice, we do not compute a power series $y_{c}(x) \in \KK[[x]]$ solution of \eqref{diff-eq} but only a truncation of $y_{c}(x)$, i.e., a finite number of terms of its expansion on the monomial basis. Given a degree bound $N$ for the rational first integral that we are searching for, the following lemma shows that computing $y_{c}(x) \mod x^{N^2+1}$, i.e., the first $N^2+1$ terms of its expansion, is enough for our purposes. 

Such an analysis of the needed precision for the power series solutions of \eqref{diff-eq} that we compute is not included in \cite{FG}. Note that this kind of strategy was already used in a polynomial factorization setting (see, e.g., \cite{Kaltof}) and in a differential equations setting (see \cite{ArCaFeGa05,BCCLSS07}). The next lemma is a small improvement of \cite[Lemma 2.4]{ArCaFeGa05}. 

\begin{Lem}\label{algebraic-solution} 
Let $\LL$ be a field of characteristic $0$ such that $\KK \subseteq \LL$. Let $\hat{y}(x) \in \LL[[x]]$ denote an algebraic power series whose minimal polynomial $M \in \LL[x,y]$ has degree at most $N$. If $\tilde{M} \in \LL[x,y]$ is a polynomial of degree at most $N$ satisfying
	$$(\star): \; \tilde{M}(x,\haty(x)) \equiv 0 \mod x^{N^2+1},$$
	then $\tilde{M}(x,\haty(x))=0$. Moreover, if $\tilde{M}$ has minimal degree in $y$ among polynomials satisfying $(\star)$, then 
	$\tilde{M}=f \, M$  for some $f \in \LL[x]$.
\end{Lem}
\begin{proof}
By definition, $M$ satisfies $(\star)$ so there exists $\tilde{M} \in \LL[x,y]$ of degree at most~$N$ satisfying $(\star)$. Let $\tilde{M}$ be such a solution of $(\star)$  and consider 
$$\mathcal{R}(x):={\rm Res}_y(M(x,y), \tilde{M}(x,y)),$$ the resultant of $M$ and $\tilde{M}$ with respect to $y$.
As there exist polynomials $S$ and $T$ in $\KK[x,y]$ such that $S\,M+T\,\tilde{M}= \mathcal{R}$, 
Relation $(\star)$ yields $\mathcal{R}(x) \equiv 0 \mod x^{N^2+1}$. By B\'ezout's theorem, we have $\deg(\mathcal{R})\leq \deg(M)\,\deg(\tilde{M}) \leq N^2$, thus $\mathcal{R} = 0$. This implies that $M$ and $\tilde{M}$ have a non-trivial common factor. Now, as $M$ is irreducible, necessarily $M$ divides $\tilde{M}$ and thus $\tilde{M}(x,\haty(x))=0$. Finally, if $\tilde{M}$ is supposed to have minimal degree in $y$ among polynomials satisfying $(\star)$, we have necessarily $\tilde{M}=f\,M$ for some $ f\in \LL[x]$ which ends the proof.
\end{proof} 

Note that in $(\star)$, the power series $\haty(x)$ can be replaced by its truncation $\haty(x) \mod x^{N^2+1}$.

\begin{Rem}\label{rmk:basis}
For a given power series $\haty(x)$, computing all the polynomials $\tilde{M}$  of degree at most $N$ satisfying $(\star)$ can be done by taking an ansatz for $\tilde{M}$ and performing linear algebra calculations (e.g., solving a system of linear equations). Consequently, computing ``all" solutions of $(\star)$ means computing ``a basis" of solutions of the linear algebra problem associated with $(\star)$. An efficient method to address this problem and to get, via a row-echelon form, a solution of $(\star)$ with minimal degree in $y$ is given in Subsection~\ref{complexity-section-min-poly} where a complexity analysis is provided.
\end{Rem}

In the sequel, we say that $\tilde{M}$ is a {\em minimal solution of $(\star)$} if it is a solution of $(\star)$ with minimal degree in $y$.

\subsection{A first  algorithm for computing rational first integrals} \label{subsect_genericalgo}

We now propose a first algorithm, based on linear algebra, for solving $({\mathcal P}_N)$. More efficient algorithms, based on this one, are given in Section \ref{section-proba_det_algo}. The strategy of this algorithm is then used in Section~\ref{sec:Darboux} for computing Darboux polynomials.\\

\noindent \underline{Algorithm  \textsf{GenericRationalFirstIntegral}}\\

\noindent \texttt{Input:}$A,B\in \KK[x,y]$ s.t. $A(0,y)\not \equiv 0$ and a bound $N \in \NN$.\\
\texttt{Output:} A  non-composite rational first integral of \eqref{eq-sys} of degree at most~$N$, or ``None".\\

\begin{enumerate}
\item \label{gen_step1} For an indeterminate $c$, compute the polynomial $y_{c}\in \KK(c)[x]$ of degree at most $ (N^2 +1)$ s.t. $y_{c}(0)=c$ and 
	$\frac{dy_{c}}{dx}\equiv \frac{B(x,y_{c})}{A(x,y_{c})} \mod x^{N^2+1}$.\\
\item \label{gen_step2} Compute all\footnote{i.e., a basis over $\KK(c)$,  see Remark \ref{rmk:basis}.} non-trivial polynomials $\tilde{M} \in \KK(c)[x,y]$ of degree $\leq N$ s.t.
 $$(\star): \tilde{M}(c,x,y_{c}(x)) \equiv 0 \mod  x^{N^2+1}.$$
If no such $\tilde{M}$ exists, then Return ``None". Else, among the solutions of  $(\star)$, pick a \emph{minimal} solution $\overline{M} \in \KK[c][x,y]$.\\
\item \label{gen_step3} Let $M$ denote the primitive part of $\overline{M}$ relatively to $y$.
\\   
Set $P(x,y):=M(0,x,y)$. \\Pick any $c_1 \in \KK$ s.t. $\frac{M(c_1,x,y)}{P(x,y)} \not\in \KK$ and set $Q(x,y):=M(c_1,x,y)$.\\
\item \label{gen_step4} If $\CD(P/Q)=0$, then Return $P/Q$. Else Return ``None''.\\
\end{enumerate}

In the above algorithm, the output ``None" means that there is no rational first integral of degree at most $N$ but it may exist a rational first integral of degree strictly greater than $N$.
\begin{Thm}\label{theorem-generic}
Algorithm  \textsf{GenericRationalFirstIntegral} is correct: either it finds a non-composite rational first integral of \eqref{eq-sys} of degree at most $N$ if it exists, or it proves that no such rational first integral exists.
\end{Thm}

To prove Theorem \ref{theorem-generic}, we shall need the following lemma.

\begin{Lem}\label{sigma_non_constant}
Consider the planar polynomial vector field \eqref{eq-sys} and assume that $A(0,y)\not \equiv 0$. 
If $F$ is a reduced rational first integral of \eqref{eq-sys}, then $F(0,y) \in \KK(y)\setminus\KK.$
\end{Lem}
\begin{proof}
Let $F=P/Q$ be a reduced rational first integral of \eqref{eq-sys}. Proceeding by contradiction,  we assume $F(0,y)=c_0 \in \KK$. Then $P(0,y)-c_0\,Q(0,y)=0$ so that $x$ divides $P(x,y)-c_0\,Q(x,y)$. Now, from Corollary~\ref{RFI-Darboux}, $P(x,y)-c_0\,Q(x,y)$ is a Darboux polynomial for $\DD$ and thus, by Lemma \ref{darboux-semi-group}, $x$ is also a Darboux polynomial for $\DD$. Consequently, we get that $x$ divides $A(x,y)$ and thus $A(0,y)=0$. This is absurd so we conclude \mbox{$F(0,y) \in \KK(y)\setminus\KK$}.
\end{proof}
\begin{proof}[Proof of Theorem \ref{theorem-generic}]
Suppose first that there exists a rational first integral of \eqref{eq-sys} of degree at most $N$. Then, without loss of generality, we can consider a reduced non-composite one $P_0/Q_0$, see Lemma~\ref{lem:RFI-alg}. Let $y_{c} \in \KK(c)[[x]]$ be the power series solution of \eqref{diff-eq} satisfying $y_{c}(0)=c$. By Proposition~\ref{lem_alg}, $y_{c}$ is a root of \mbox{$\lambda \, P_0 -\mu \,Q_0$}, where  $\lambda=Q_0(0,c)$ and $\mu=P_0(0,c)$. As $P_0/Q_0$ is non-composite, it follows that $\lambda \, P_0- \mu \, Q_0$ is irreducible in $\KK(c)[x,y]$. Indeed, by Lemma~\ref{sigma_non_constant}, the constant $\mu / \lambda$ belongs to $\KK(c)\setminus \KK$ and thus, from Lemma~\ref{lem:remarkable_values}, we can find  $c_0 \in \KK$ such that $(Q_0(0,c_0):P_0(0,c_0)) \not \in  \sigma(P_0,Q_0)$. Consequently $\lambda \, P_0 - \mu \, Q_0$ is a minimal polynomial of $y_{c}$. In Step~(\ref{gen_step1}), we compute the first $N^2+1$ terms of $y_{c}$. Now, in Step~(\ref{gen_step2}), if there exists a solution $\tilde{M}  \in \KK(c)[x,y]$ of $(\star)$ of degree at most $N$, then, Lemma \ref{algebraic-solution} applied with $\LL=\KK(c)$ implies that $\overline{M}=f\,(\lambda \, P_0-\mu \,Q_0)$ with $f \in \KK(c)[x]$, where $\overline{M} \in \KK[c][x,y]$ is defined in Step (\ref{gen_step2}). Therefore, taking the primitive part of $\overline{M}$ with respect to $y$, in Step~(\ref{gen_step3}), we have $M=g\,(\lambda \, P_0-\mu \, Q_0)$ for some $g \in \KK[c]$. Now, if $P$ and $Q$ denote the polynomials defined in Step~(\ref{gen_step3}) of the algorithm, we necessarily have:
$$\dfrac{P}{Q}=\dfrac{\alpha \, P_0+\beta \, Q_0}{\delta \, P_0 + \gamma \, Q_0} \in \KK(x,y)\setminus \KK, \textrm{ where } \alpha, \beta, \delta , \gamma\in \KK.$$
As $P_0/Q_0$ is a non-composite rational first integral, we deduce that $P/Q$ is also a non-composite rational first integral. Thus, we have $\DD(P/Q)=0$ in Step~(\ref{gen_step4}) and the algorithm returns a correct output.\\
Now suppose that \eqref{eq-sys} has no rational first integral of degree at most $N$.  In Step (\ref{gen_step4}), the test $\DD(P/Q)=0$ guarantees to return a correct output. In Step~(\ref{gen_step2}), we can have an early detection of this situation. Indeed by Proposition~\ref{lem_alg}, if $(\star)$ has no non-trivial solution, then 
we deduce that \eqref{eq-sys} has no rational first integral of degree at most $N$.
\end{proof}

This algorithm fits the first part of our goal as it is entirely based on linear operations: we do not need to solve quadratic equations (see Section~\ref{complexity-section}). However, it is not yet very efficient in practice because computations are done over $\KK(c)$. For example, in the first step, a direct calculation shows that, for $n \geq 1$, the coefficient of $x^n$ in the power series solution $y_{c}$ of \eqref{diff-eq} satisfying $y_{c}(0)=c$ is generically a rational function in $c$ of degree $(2n-1)\,d$, whose denominator is generically $A(0,c)^{2\,n-1}$. 
In what follows,  we accelerate things by using only computations over $\KK$ instead of computations in $\KK(c)$.

\section{Efficient algorithms for computing rational first integrals}\label{section-proba_det_algo}
\subsection{A probabilistic algorithm} \label{subsec:proba}
In this section, we present an efficient probabilistic algorithm of Las Vegas type for solving $({\mathcal P}_N)$. The approach is similar to the one used in the previous section. \\

\noindent \underline{Algorithm  \textsf{ProbabilisticRationalFirstIntegral}}\\

\noindent \texttt{Input:}$A,\,B\in \KK[x,y]$ s.t. $A(0,y)\not \equiv 0$, two elements $c_1,\,c_2\in \KK$ s.t. $A(0,c_i)\neq 0$ for $i=1,\,2$, and a bound $N \in \NN$.\\
\texttt{Output:} A non-composite rational first integral of \eqref{eq-sys} of degree  at most $N$, ``None" or ``I don't know''.\\

\begin{enumerate}
\item \label{step_proba_1} For $i=1,2$ do:\\
\begin{itemize}
\item[(1a)] \label{step_1a} Compute  
		$y_{c_i} \in  \KK[x]$ 
		of degree at most $(N^2 +1)$ s.t. $y_{c_i}(0)=c_i$, and \\
	$\frac{dy_{c_i}}{dx}\equiv \frac{B(x,y_{c_i})}{A(x,y_{c_i})} \mod x^{N^2+1}$.\\
\item[(1b)] \label{step_1b} Compute all non-trivial polynomials $\tilde{M_i} \in \KK[x,y]$ of degree $\leq N$ s.t. 
$$(\star): \tilde{M_i}(x,y_{c_i}(x)) \equiv 0 \mod  x^{N^2+1}.$$
\item[(1c)] \label{step_1c} If no such $\tilde{M_i}$ exists, then Return ``None''.\\ Else let $M_i \in \KK[x,y]$ be the primitive part relatively to $y$ of a minimal solution of $(\star)$.\\
\item[(1d)] If $i=1$, then while ($M_1(0,c_2)=0$ or $A(0,c_2) = 0$) do $c_2=c_2+1$.\\
\end{itemize}
\item \label{step-proba-4}If $\DD(M_1/M_2)=0$, then Return $M_1/M_2$. Else Return [``I don't know'',$[c_2]$].
\end{enumerate}

\begin{Thm} \label{theorem-probabilistic}
Algorithm  \textsf{ProbabilisticRationalFirstIntegral} terminates and satisfies the following properties:
\begin{itemize}
\item If it returns $M_1/M_2$, then it is a non-composite rational first integral of \eqref{eq-sys} of degree at most $N$.
\item If it returns ``None'', then there is no rational first integral of \eqref{eq-sys} of degree at most $N$.
\item If \eqref{eq-sys} admits a non-composite rational first integral $P/Q$ of degree at most $N$ and $(Q(0,c_i):P(0,c_i)) \not \in \sigma(P,Q)$ for $i=1,\,2$, then the algorithm returns  a non-composite rational first integral of \eqref{eq-sys} of degree at most $N$. 
\end{itemize}
\end{Thm}

\begin{proof}
Let us first prove that the algorithm terminates. This follows directly from the fact that the \textsf{while} loop in Step (1d) terminates after at most $N+d+1$ steps. Indeed, we just have to avoid the roots of the product $M_1(0,y)\,A(0,y)$ which a univariate polynomial of degree less than $N+d$. It thus remains to check that it is a non-zero polynomial, i.e., $M_1(0,y) \not \equiv 0$. If $M_1(0,y) \equiv 0$, then $x$ divides $M_1$. As $M_1$ is the primitive part with respect to $y$ of a minimal solution of $(\star)$, this would imply that $M_1(x,y)=x$ and thus $M_1(x,y_{c_1}(x)) \not \equiv 0 \mod  x^{N^2+1}$ which is a contradiction.\\
Now, if the algorithm returns $M_1/M_2$, then the test in Step (2) ensures that $\CD(M_1/M_2)=0$ and, by construction, $M_1/M_2$ is clearly of degree at most $N$. Furthermore, $M_1/M_2$ is non-composite. Indeed, if $M_1/M_2$ is composite, then at least one of the $M_i$'s is reducible and thus it can not be the primitive part with respect to $y$  of a minimal solution of $(\star)$. Finally Step (1d) certifies that $M_1/M_2 \not \in \KK$. Indeed, $M_2$ satisfies $M_2(0,c_2)=0$, thus if $M_2=k\,M_1$ with $k \in \KK$, then either $k=0$ or $M_1(0,c_2)=0$ which is not possible thanks to Step (1d). We have then proved that $M_1/M_2$ is a non-composite rational first integral of \eqref{eq-sys} of degree at most $N$. \\
If the algorithm returns ``None'' in Step (1c), then by Proposition~\ref{lem_alg}, \eqref{eq-sys} has no rational first integral of degree at most $N$.\\
Assume finally that \eqref{eq-sys} admits a non-composite rational first integral $P/Q$ of degree at most $N$ and that $(Q(0,c_i):P(0,c_i)) \not \in \sigma(P,Q)$ for $i=1,\,2$. Then the same strategy as the one used in the proof of Theorem~\ref{theorem-generic} shows that our algorithm returns a non-composite rational first integral of \eqref{eq-sys}.\\
\end{proof}
\begin{Prop}
Let $\Omega$ be a (finite) subset of $\KK$ of cardinal $|\Omega|$ greater than $N\,(\mathcal{B}(d)+1)$ and assume that, in Algorithm \textsf{ProbabilisticRationalFirstIntegral}, $c_1$ and $c_2$ are chosen independently and uniformly at random in $\Omega$. Then, if \eqref{eq-sys} admits a  rational first integral of degree at most $N$, Algorithm \textsf{ProbabilisticRationalFirstIntegral} returns a non-composite rational first integral of \eqref{eq-sys} of degree at most~$N$ with probability at least $\left(1-\frac{N\,(\mathcal{B}(d)+1)}{|\Omega|}\right)$.
\end{Prop}

\begin{proof}
It is a straightforward application of Lemma~\ref{lem:remarkable_values}, Theorem~\ref{theorem-probabilistic} and Zippel-Schwartz's lemma (see \cite[Lemma 6.44]{GG}).
\end{proof}

In fact, the ``practical" probability will be much better. Indeed, the elements $(\lambda:\mu)$ of the spectrum may be rational or algebraic and hence, the constants $c$ such that $(Q(0,c):P(0,c)) \in \sigma(P,Q)$ will generally be algebraic. So, if the $c_i$'s are chosen to be rational in the input, then the ``bad" values of the $c_i$'s will generally be in very small number. This fact is widely confirmed by experiments.\\

Now, we study all the different situations that can occur and the corresponding output given by the algorithm  \textsf{ProbabilisticRationalFirstIntegral}: 
\begin{enumerate}
\item \eqref{eq-sys} has a non-composite rational first integral $P/Q$ of degree at most $N$.
\begin{enumerate}
\item If $(Q(0,c_1):P(0,c_1)) \not \in \sigma(P,Q)$, and $(Q(0,c_2):P(0,c_2)) \not \in \sigma(P,Q)$ then in this situation the algorithm returns a non-composite rational first integral.
\item  
Now, we study the opposite situation: $(Q(0,c_1):P(0,c_1))  \in \sigma(P,Q)$ or $(Q(0,c_2):P(0,c_2)) \in \sigma(P,Q)$. If the algorithm computes $M_1$ and $M_2$ but $M_1/M_2$ is not a rational first integral, then it returns ``I don't know". A first example where this case is encountered is given in Subsection~\ref{subsec:probaex}. Furthermore, we may be unlucky enough to choose two bad values of the $c_i$'s, i.e., \mbox{$(Q(0,c_i):P(0,c_i)) \in \sigma(P,Q)$} for $i=1,\,2$. For example if we consider $$A(x,y)=-4\,{x}^{3}+4\,x{y}^{2}+6\,{x}^{2}-2\,{y}^{2}-2\,x, \; B(x,y)=-4\,{x}^{2}y+4\,{y}^{3}+4\,x\,y-2\,y,$$
then \eqref{eq-sys} has a non-composite rational first integral $P/Q$ of degree $2$, where $P(x,y)=(y-x)\,(y-x+1)$ and $Q(x,y)=(y+x)\,(y+x-1)$. But if we choose $c_1=-1$ and $c_2=1$, then we will construct two Darboux polynomials  $M_1(x,y)=y-x+1$ and $M_2(x,y)=y+x-1$ of degree only $1$ that are minimal polynomials of $y_{c_1}$ and $y_{c_2}$. As $\deg(M_1/M_2)$ is strictly smaller than $\deg(P/Q)$, we obtain $\DD(M_1/M_2) \neq 0$ and the algorithm returns ``I don't know". 
\end{enumerate}
\item \eqref{eq-sys} does not have a rational first integral with degree at most $N$.
\begin{enumerate}
\item If $(\star)$ has no non-trivial solutions, then the algorithm returns ``None". 
\item If $(\star)$ has non-trivial solutions, then the algorithm returns ``I don't know". This situation can occur for example when:
\begin{itemize}
\item \eqref{eq-sys} has no rational first integral but it has Darboux polynomials and the choice of $c_1$ and $c_2$ gives two Darboux polynomials. For an example of a derivation without rational first integral but with Darboux polynomials, see \cite[Remark 15]{Ch}.
\item \eqref{eq-sys} has a rational first integral with degree bigger than the given bound $N$.\\
For example, consider the derivation $\CD=(x+1)\,\frac{\partial}{\partial x} - y \, \frac{\partial}{\partial y}$ and the degree bound $N=1$. In this situation, the differential equation is $(E): \frac{dy}{dx}=\frac{-y}{x+1}$ which admits $y_{c}(x)=\frac{c}{1+x}$ as solution. We set $M(x,y)=\alpha+\beta x +\gamma y$, and then $M(x,y_{c}(x))=0 \mod x^2$ gives $M(x,y)=\gamma(-c+cx+y)$. However, $M(x,y_{c}(x))=x^2 \mod x^3$, thus $y_{c}(x)$ is not a root of $M$. Here $\CD$ admits the rational first integral $y\,(x+1)$ so if we set $N=2$ in the input, our algorithm returns a non-composite rational first integral of degree $2$. In this case we compute $y_{c}(x) \mod x^5$.
\end{itemize}
\end{enumerate}
\end{enumerate}

\begin{Rem}
The bivariate polynomials $\tilde{M_i}$'s computed in Step~(1b) have total degree at most $N$ so they have \mbox{$(N+1)(N+2)/2$} coefficients. Note that, if we assume $N\geq 3$, then we have $N^2+1 \geq (N+1)(N+2)/2$. It is tempting to try to compute the $\tilde{M_i}$'s using only, say,
$(N+1)(N+2)/2+2$ terms of the power series. This will make the computation a little bit faster, but then the method becomes only a nice heuristic and may fail.
\end{Rem}

\subsection{A deterministic algorithm}
Algorithm \textsf{ProbabilisticRationalFirstIntegral} is now turned into a deterministic algorithm. The idea is that if a rational first integral with degree at most $N$ exists then, if we run at most $N\,(\mathcal{B}(d)+1)+1$ times  \textsf{ProbabilisticRationalFirstIntegral}, we will get a non-composite rational first integral of degree at most $N$.\\

\noindent \underline{Algorithm  \textsf{DeterministicRationalFirstIntegral}}\\

\noindent \texttt{Input:}$A,\,B\in \KK[x,y]$ s.t. $A(0,y) \not \equiv 0$ and a bound $N \in \NN$.\\
\texttt{Output:} A non-composite rational first integral of \eqref{eq-sys} of degree $\leq N$ or ``None".\\

\begin{enumerate}
\item Let $\Omega:=\emptyset$.\\
\item While $|\Omega| \leq 2\,N\,(\mathcal{B}(d)+1) +2$ do \label{det-step2}\\
\begin{enumerate}
\item Choose two random elements $c_1,\,c_2 \in \KK \setminus \Omega\,$ s.t. $c_1 \neq c_2$ and \mbox{$A(0,c_i)\neq 0$} for $i=1,\,2$. 
\item $F:=\textsf{ProbabilisticRationalFirstIntegral}(A,B,(c_1,c_2),N)$.
\item If $F=$``None", then Return ``None". 
\item Else if $F=$[``I don't know'',$[e_2]$], then $\Omega:=\Omega \cup \{c_1,e_2\}$ and go to Step~(\ref{det-step2}).
\item \label{det-step4} Else Return $F$. \\
\end{enumerate}
\item Return ``None''.
\end{enumerate}

\begin{Thm} \label{correct_detalgo}
Algorithm  \textsf{DeterministicRationalFirstIntegral} is correct: it returns a rational first integral of degree at most $N$ if and only if it exists, and it returns ``None'' if and only if there is no rational first integral of degree at most $N$.
\end{Thm}
\begin{proof}
Assume that \eqref{eq-sys} has a non-composite rational first integral $P/Q$ with degree at most $N$. If $F=$``I don't know''  in Step (\ref{det-step2}), then from Theorem~\ref{theorem-probabilistic}, at least one of the $c_i$'s satisfies $(Q(0,c_i):P(0,c_i)) \in \sigma(P,Q)$. The number of such ``bad" values of the $c_i$'s is bounded by $N\,(\mathcal{B}(d)+1)$ by Lemma \ref{lem:remarkable_values}. Hence if we repeat \textsf{ProbabilisticRationalFirstIntegral} at least $N\,(\mathcal{B}(d)+1)+1$ times, then we will get a good pair $(c_1,c_2)$ and by Theorem~\ref{theorem-probabilistic}, the probabilistic algorithm will then return a non-composite rational first integral of degree at most $N$.\\
Now assume that \eqref{eq-sys} has no rational first integral of degree at most $N$. Then by Theorem \ref{theorem-probabilistic}, \textsf{ProbabilisticRationalFirstIntegral} returns ``None" or ``I don't know".\\
If in Step~(\ref{det-step2}),  $F$=``None", then we have a correct output. Now if $F$=``I don't know", then the algorithm uses again \textsf{ProbabilisticRationalFirstIntegral} with new values of the $c_i$'s and, after at most $N(\mathcal{B}(d)+1)+1$ trials, it returns ``None" which is the correct output.
\end{proof}

\section{Complexity analysis and algorithmic issues} \label{complexity-section} 

In this section, we describe how the different steps of algorithms \textsf{ProbabilisticRationalFirstIntegral} and \textsf{DeterministicRationalFirstIntegral} can be performed efficiently and we study their arithmetic complexities. 
For the complexity issues, we focus on the dependency on the degree bound $N$ and we recall that we assume that $N\geq d$, where $d=\max(\deg(A),\deg(B))$ denotes the degree of the polynomial vector field. More precisely, we suppose that $d$ is fixed and $N$ tends to infinity.\\

All the complexity estimates are given in terms of arithmetic operations in $\KK$. 
We use the notation $f \in \bigOsoft(g)$: roughly speaking, it means that  $f$ is in $\bigO(g\,\log^m(g))$ for some $m \geq 1$. For a precise definition, see \cite[Definition 25.8]{GG}. We suppose that the Fast Fourier Transform can be used so that two univariate polynomials with coefficients in $\KK$ and degree bounded by $D$ can be multiplied in  $\bigOsoft(D)$, see \cite{GG}. We further assume that  two matrices of size $n$ with entries in $\KK$ can be multiplied using $\bigO(n^\omega)$,  where $2 \leq \omega \leq 3$ is the matrix multiplication exponent, see \cite[Ch.~12]{GG}. 
We also recall that a basis of solutions of a linear system composed of $m$ equations and $n \leq m$ unknowns over $\KK$ can be computed using $\bigO(m\,n^{\omega-1})$ operations in $\KK$, see \cite[Chapter 2]{BP}.

\subsection{Computation of a regular point} \label{subsec:regular}
In the algorithms given in the previous sections, we have to choose a regular point for the differential equation \eqref{diff-eq}, i.e., a point $x_0$ satisfying $A(x_0,y) \not \equiv 0$. To achieve this, we can start from the point $x_0=0$, evaluate $A(x,y)$ at $x=x_0$. If $A(x_0,y) \not \equiv 0$, then we are done. Else, we shift  $x_0$ by one to get $x_0=1$ and we iterate the 
process. Note that the number of iterations is at most $d$. Consequently, this step can be performed by evaluating 
$d$ polynomials (namely the coefficients of $A(x,y)$ viewed as polynomials in the variable $y$) of degree bounded by $d$ 
at $d$ points ($x_0=0,\,1,2,\dots,d-1$). This can thus be done in $\bigOsoft(d^2)$ 
arithmetic operations, see \cite[Corollary 10.8]{GG}. This is why, in our algorithms, we always suppose, at neglectable cost and without loss of generality, that $A(0,y) \not \equiv 0$.

\subsection{Power series solutions of \eqref{diff-eq} }\label{complexity-section-pow-ser}
In Step~(\ref{step_proba_1}) of the algorithm  \textsf{ProbabilisticRationalFirstIntegral}, we  compute the $N^2+1$ first terms of the power series solution of \eqref{diff-eq} satisfying a given initial condition. Using the result of Brent and Kung (see \cite[Theorem 5.1]{BK}) based on formal Newton iteration, this can be done using $\bigOsoft(d \, N^2)$ arithmetic operations, see also \cite{BCLOSSS}.
\subsection{Guessing the minimal polynomial of an algebraic power series}\label{complexity-section-min-poly}

We shall now give a method for solving Problem $(\star)$ in Step~(1b) of Algorithm \textsf{ProbabilisticRationalFirstIntegral}. The problem is the following: 
given the first $N^2+1$ terms of a power series $\haty(x)$, find (if it exists), a bivariate polynomial $M \in \KK[x,y]$,  with minimal degree in $y$, such that $M(x,\haty(x))\equiv0 \mod x^{N^2+1}$. This can be handled by an undetermined coefficients approach as follows:\\

\noindent \underline{Algorithm  \textsf{GuessMinimalPolynomial}}\\

\noindent \texttt{Input:} A \emph{polynomial} $\haty \in \KK[x]$ s.t. $\deg(\haty) \leq (N^2+1)$, with $N  \in \NN$.\\
\noindent \texttt{Output:} A minimal solution of $(\star)$ with degree $\leq N$ or ``None''.\\

\begin{enumerate} 
\item Let $M(x,y)=\sum_{i=0}^N \, \left( \sum_{j=0}^{N-i} m_{i,j} \, x^j \right) \, y^i$ be an ansatz for the bivariate polynomial that we are searching for.\\
\item Construct the linear system $(\mathcal{L})$ for the $m_{i,j}$'s given by:  
$$M(x,\haty(x))=\sum_{i=0}^N \, \left( \sum_{j=0}^{N-i} m_{i,j} \, x^j \right) \, \haty(x)^i \equiv 0 \mod x^{N^2+1}.$$
\item If $(\mathcal{L})$ does not have a non-trivial solution, then Return ``None". \\
\item \label{rowechelon} Else compute a row-echelon form of a basis of solutions of $(\mathcal{L})$ to find a solution $M(x,y)$ of minimal degree in $y$ and Return it.
\end{enumerate}

\begin{Prop}
Algorithm  \textsf{GuessMinimalPolynomial} is correct. If we suppose that $N\geq 3$, then it uses at most $\bigOsoft(N^{2 \omega})$ arithmetic operations in $\KK$.
\end{Prop}

\begin{proof}
The correctness of the algorithm is straightforward. Let us study its arithmetic complexity. To construct the linear system $(\mathcal{L})$, we have to compute $\haty(x)^i \mod x^{N^2+1}$ for $i=0,\ldots,N$. This can be done in $\bigOsoft(N^3)$ arithmetic operations. The linear system $(\mathcal{L})$ has $N^2+1$ equations and $(N+1)\,(N+2)/2=\bigO(N^2)$ unknowns $m_{i,j}$'s. Note that we assume $N \geq 3$ so that $N^2+1 \geq (N+1)\,(N+2)/2$. It can thus be solved 
using $\bigO(N^2\,(N^2)^{\omega-1})$ operations. Finally, in Step~(\ref{rowechelon}), the row-echelon form can be computed using at most $\bigOsoft(N^{2\,\omega})$ arithmetic operations (see \cite[Chapter 3]{BP}) since the dimension of a basis of solutions of $(\mathcal{L})$ does not exceed $\bigO(N^2)$, which ends the proof. 
\end{proof}


\subsection{Total cost of our algorithms}

\begin{Thm} \label{compl_proba}
Algorithm \textsf{ProbabilisticRationalFirstIntegral} uses at most $\bigOsoft(N^{2 \omega})$ arithmetic operations in $\KK$, when $N$ tends to infinity and $d$ is fixed.
\end{Thm}
\begin{proof}
In Subsection~\ref{complexity-section-pow-ser}, we have seen that Step (1a) can be performed in at most $\bigOsoft(d\,N^{2})$ arithmetic operations. Then, using Algorithm \textsf{GuessMinimalPolynomial}, Step~(1b) 
can be performed in  $\bigOsoft(N^{2\,\omega})$ operations in $\KK$, see Subsection~\ref{complexity-section-min-poly}. 
In Step~(1c), we have to compute the primitive part relatively to $y$ of a minimal solution of $(\star)$. This reduces to computing $N$ $\gcd$'s of univariate polynomials of degree at most $N$ which can be done in $\bigO(N^3)$ operations in $\KK$ (and even faster using half-gcd techniques). In Step~(1d), we must avoid the roots of $M_1(0,y)\,A(0,y)$ thus we need to run the \textsf{while} loop at most $d+N+1$ times. In this loop we evaluate univariate polynomials with degree at most $d$ and $N$, thus it uses at most $\bigOsoft( (d+N)^2)$ arithmetic operations. Finally, we test if $\DD(M_1/M_2)=0$ which costs $\bigOsoft((d+N)^2)$ arithmetic operations since $N\geq d$. Indeed, we multiply bivariate polynomials of degree at most $N$ and we add bivariate polynomials of degree at most $d+2N-1$. 
\end{proof}
\begin{Cor}
The deterministic algorithm \textsf{DeterministicRationalFirstIntegral} can be done using at most $\bigOsoft(d^2\,N^{2 \omega+1})$ 
	arithmetic operations, when $N$ tends to infinity and $d$ is fixed.
\end{Cor}

In the previous statement, even if $d$ is fixed, we mention it in the complexity in order to emphasize on the number of iterations of the probabilistic algorithm.
\begin{proof}
This estimate is straightforward from Theorem~\ref{compl_proba} since Algorithm \textsf{DeterministicRationalFirstIntegral} calls at most $N\,(\mathcal{B}(d)+1)+1$ times the algorithm \textsf{ProbabilisticRationalFirstIntegral}.
\end{proof}



\subsection{Faster heuristic using Pad\'e-Hermite approximation}\label{complexity-section-pade}

The algorithm \textsf{GuessMinimalPolynomial} developed in Subsection~\ref{complexity-section-min-poly} uses an undetermined coefficients method to compute a minimal solution of $(\star)$ in Step (1b) of Algorithm \textsf{ProbabilisticRationalFirstIntegral}. It consists in finding (if it exists) the minimal polynomial of a power series. In the present section, we propose another approach to solve that problem using Pad\'e-Hermite approximation, see \cite{BL}.\\ 
Indeed, the problem of computing a bivariate polynomial annihilating a power series can be handled by means of computing a  Pad\'e-Hermite approximant, see \cite{Sh,Sh1}. More precisely, given a power series $\haty(x)$, if there exists a bivariate polynomial $M$ of degree $N$ such that $M(x,\haty(x))=0$, then the coefficients of the powers of $y$ are a Pad\'e-Hermite approximant of type $(N,N-1,\ldots,0)$ of the vector of power series $(1,\haty(x),\ldots,\haty(x)^N)^T$. Computing such a Pad\'e-Hermite approximant provides a polynomial $\tilde{M}$ satisfying $\tilde{M}(x,\haty(x)) \equiv 0 \mod x^{\sigma}$ where 
$\sigma =N\,(N+1)/2+N-1$. Unfortunately $\sigma < N^2+1$ so that 
 we have no way to ensure, using Lemma \ref{algebraic-solution}, that the Pad\'e-Hermite approximant computed 
satisfies $\tilde{M}(x,\haty(x))=0$. Consequently, using this method to compute the $M_i$'s in Step (1b) of Algorithm \textsf{ProbabilisticRationalFirstIntegral} only provides a heuristic.

\begin{Prop}
Using Pad\'e-Hermite approximation in Step (1b), Algorithm \textsf{ProbabilisticRationalFirstIntegral} becomes a heuristic
 for computing a non-composite rational first integral of \eqref{eq-sys} of degree at most $N$ using only $\bigOsoft(N^{\omega+2})$ arithmetic operations. 
\end{Prop} 

\begin{proof}
Beckermann-Labahn's algorithm (see \cite{BL}) computes a Pad\'e-Hermite approximant of type $(N,N-1,\ldots,1)$ of the vector of power series $(1,\haty(x),\ldots,\haty(x)^N)^T$ in $\bigOsoft(N^{\omega}\,\sigma)$ arithmetic operations, where $\sigma =N\,(N+1)/2+N-1$. Using  the proof of Theorem~\ref{compl_proba}, we obtain the desired complexity estimate.
\end{proof}

\section{Implementation and experiments} \label{sec_impl_exp}

The algorithms developed in the previous sections have been implemented in a Maple package called {\sc RationalFirstIntegrals}. It is available with some examples at \url{http://www.ensil.unilim.fr/~cluzeau/RationalFirstIntegrals.html}. \\


Our implementation of the heuristic proposed in Subsection~\ref{complexity-section-pade} is called \textsf{HeuristicRationalFirstIntegral}. It uses the {\sc gfun} package~\cite{SaZi94}\footnote{\url{http://perso.ens-lyon.fr/bruno.salvy/?page_id=48}} and more precisely its \textsf{seriestoalgeq} command to search for a bivariate polynomial annihilating the power series computed using Pad\'e-Hermite approximation. \\

We shall now illustrate our implementation and give some timings\footnote{All the computations were made on a 2.7 GHz  Intel Core i7}. 
 
\subsection{Comparison to previous methods}

We start by comparing our implementation \textsf{DeterministicRationalFirstIntegral} to two previous methods, namely:
\begin{enumerate}
\item the {\em naive} approach which consists in using the method of undetermined coefficients to search for two polynomials $P$ and $Q$ of degree at most $N$ satisfying $\DD(P)\,Q-P\,\DD(Q)=0$. This implies solving a system of quadratic equations in the coefficients of $P$ and $Q$. In our implementation, we use the \textsf{solve} command of Maple to solve the quadratic system,
\item the approach developed in \cite{Ch} based on the {\em ecstatic curve}. 
\end{enumerate}

Consider the planar polynomial vector field given by $A(x,y)=-7\,x+22\,y-55$ and $B(x,y)=-94\,x+87\,y-56$ which has no rational first integral of degree less than $6$. The following table compares the timings (in seconds) of the different implementations for proving the non-existence of a rational first integral of degree less than $N=2,\ldots,6$.  
 
\begin{center}
\begin{tabular}{|c||c|c|c|} \hline
\backslashbox{$N$}{Method} &\textsf{DeterministicRFI} &  Ecstatic curve  & Naive method \\ \hline \hline
$2$   & 0.043 & 0.003 & 0.257  \\ \hline  
$3$   & 0.006 & 0.024 & 0.043 \\ \hline
$4$   & 0.016 & 3.310 & 4.438 \\ \hline 
$5$    & 0.041 & 74.886 & 16.202 \\ \hline
$6$  & 0.140 &  1477.573 & 88.482 \\ \hline 
\end{tabular}
\end{center}

If we now consider the vector field given by the polynomials $A(x,y)=x+2$ and $B(x,y)=-x^2-2\,x\,y-y^2-2\,x-y-2$ which admits the rational first integral $\frac{x^2+x\,y-2}{x+y+1}$ of degree $2$, we obtain the following timings (in seconds) depending on the degree bound $N$ given in the input:\\

\begin{center}
\begin{tabular}{|c||c|c|c|} \hline
\backslashbox{$N$}{Method} &\textsf{DeterministicRFI} &  Ecstatic curve & Naive method \\ \hline \hline
$2$   & 0.012 & 0.003 &  0.137 \\ \hline  
$3$   & 0.019 & 0.019 & 1.961 \\ \hline
$4$   & 0.042 & 0.283 & 5.398  \\ \hline 
$5$    & 0.087 & 1.662 & 22.580  \\ \hline
$6$  & 1.276 & 8.491 & 80.491 \\ \hline 
\end{tabular}
\end{center}

In the latter table, the timings indicated for the ``ecstatic curve method" correspond only to the computation of the $N$th ecstatic curve (which will be zero in all cases as there exists a rational first integral of degree $2$) and not to the entire computation of a rational first integral of degree at most $N$ which requires some more computations, see \cite[Subsection 5.2]{Ch} for more details. The output of the two other methods consists in a rational first integral of degree $2$.\\

The timings presented in this subsection illustrate that our implementation of \textsf{DeterministicRationalFirstIntegral} is significantly faster than our implementations of the two previous methods considered both in the case where there exists no rational first integral and in the case where  there exists a rational first integral. This is coherent with the complexity analysis developed in this article. 

\subsection{Generic polynomial vector fields}

If we choose at random two bivariate polynomials $A$ and $B$, then the associated planar polynomial vector field has generically no rational first integral. In this subsection, we show that our implementation of  \textsf{DeterministicRationalFirstIntegral} detects quickly the non-existence of rational first integral of these generic polynomial vector fields. The following table of timings is constructed as follows: for $d=1,\ldots,10$, we generate two randomized bivariate polynomials $A$ and $B$ of degree $d$ using the \textsf{randpoly} command of Maple and we check that $A(0,y) \not \equiv 0$. Then, we run \textsf{DeterministicRationalFirstIntegral} with $N=d,\ldots,10$ and we indicate the timings (in seconds) for detecting the non-existence of a rational first integral of degree at most $N$, i.e.,  for returning the output ``None".
\begin{center}
\begin{tabular}{|c||c|c|c|c|c|c|c|c|c|c|} \hline
\backslashbox{$d$}{$N$} & $1$  & $2$  & $3$ & $4$ & $5$ & $6$ & $7$ & $8$ & $9$ & $10$   \\ \hline \hline
$1$   & 0.007& 0.043& 0.006& 0.017& 0.049& 0.157& 0.294& 1.054& 2.275& 5.010 \\ \hline  
$2$  & - & 0.044& 0.008 & 0.022& 0.070& 0.214& 0.588& 1.055& 4.644& 11.249 \\ \hline
$3$   & - & - & 0.008& 0.161& 0.084& 0.273& 0.498& 1.458&  5.267& 7.676  \\ \hline 
$4$   &  - & - & - & 0.043& 0.107& 0.577& 0.298&  2.354& 9.688& 9.10  \\ \hline
$5$  & - & -&-&-& 0.133& 0.466& 0.812& 2.439& 3.054& 8.012 \\ \hline 
$6$   &-&-&-&-&-& 0.533& 0.967&  1.557& 5.946& 5.031  \\ \hline  
$7$   & -&-&-&-&-&-& 0.663& 1.323& 4.724& 9.834 \\ \hline
$8$   & -&-&-&-&-&-&-& 2.468& 3.551& 5.756 \\ \hline 
$9$    &  -&-&-&-&-&-&-&-& 7.898& 18.127 \\ \hline
$10$  & -&-&-&-&-&-&-&-&-& 18.295  \\ \hline 
\end{tabular}
\end{center}

\subsection{Our probabilistic algorithm may fail} \label{subsec:probaex}

We now illustrate one particular case where our probabilistic algorithm \textsf{ProbabilisticRationalFirstIntegral} fails 
and returns ``I don't know". Consider the polynomial vector field given by the polynomials $$A(x,y)={x}^{6}-{x}^{5}+2\,{x}^{4}y-{x}^{4}+2\,{x}^{3}y-{x}^{2}{y}^{2}+x{y}^{2
}-{x}^{2}-2\,xy+{y}^{2}+x-2\,y+1,$$ 
$$B(x,y)=-{x}^{6}+2\,{x}^{5}y-3\,{x}^{4}y+4\,{x}^{3}{y}^{2}+3\,{x}^{4}-4\,{x}^{
3}y+3\,{x}^{2}{y}^{2}-2\,x{y}^{3}+{y}^{3}-3\,{x}^{2}+2\,xy-{y}^{2}-y+1,
$$
which admits the rational first integral of degree $4$ $$F(x,y)=\frac{P(x,y)}{Q(x,y)}={\frac { \left( y-x \right)  \left( {x}^{2}+y-1 \right) }{{x}^{4}+{y
}^{2}-1}}.
$$
If we run \textsf{ProbabilisticRationalFirstIntegral} with the bound $N=4$ and $c_1=0$ or $c_2=0$ in the input, then 
we get ``I don't know". The reason why our algorithm fails is that $(Q(0,0) :P(0,0)) =(-1:0) \in \sigma(P,Q)$ since \mbox{$-P(x,y)=- \left( y-x \right)  \left( {x}^{2}+y-1 \right)$} is a reducible polynomial (and also a polynomial of degree less than $N=4$). Of course, running \textsf{ProbabilisticRationalFirstIntegral} 
with values of $c_1$ and $c_2$ such that $(Q(0,c_i) :P(0,c_i)) \not \in \sigma(P,Q)$ for $i=1,\,2$ provides the correct output, i.e., a rational first integral of degree $N=4$; see the explanations at the end of Subsection~\ref{subsec:proba}. The deterministic algorithm \textsf{DeterministicRationalFirstIntegral} calls recursively  \textsf{ProbabilisticRationalFirstIntegral}  and exploits the fact that there only exists a finite number of such bad values of the $c_i$'s. So in this example, it returns correctly a rational first integral of degree $N=4$.

\subsection{Examples from the work of Ferragut and Giacomini}

Let us consider  \cite[Example 1]{FG},  where we have 
$$A(x,y)=6\,{x}^{4}+27\,{x}^{3}-9\,{x}^{2}y+42\,{x}^{2}-24\,xy+4\,{y}^{2}+21
\,x-7\,y+4,$$ and $$B(x,y)=18\,{x}^{4}+99\,{x}^{3}-39\,{x}^{2}y+2\,x{y}^{2}+150\,{x}^{2}-80\,xy
+12\,{y}^{2}+71\,x-21\,y+12.$$ 
A first integral  of degree $4$ was found in $12$ seconds using their algorithm (see \cite{FG}) which was a notable improvement on previous methods. In comparison, running \textsf{HeuristicRationalFirstIntegral} (or \textsf{DeterministicRationalFirstIntegral}) with $N=4$ we get such a rational first integral $F=P/Q$ in $0.022$ seconds, where
\begin{eqnarray*}
P(x,y)=-216\,{x}^{4}+144\,{x}^{3}\,y-24\,{x}^{2}\,{y}^{2}-720\,{x}^{3}+528\,{x}^{2}\,y-144\,x\,{y}^{2} \\ +16\,{y}^{3} 
+8868\,{x}^{2}+432\,x\,y-72\,{y}^{2}+28548\,x-9516\,y+9580,
\end{eqnarray*}
and
\begin{eqnarray*}
Q(x,y)=513\,{x}^{4}-342\,{x}^{3}\,y+57\,{x}^{2}\,{y}^{2}+1710\,{x}^{3}-1254\,{x}^{2}\,y+342\,x\,{y}^{2} \\ 
 -38\,{y}^{3}-10869\,{x}^{2}-1026\,x\,y+171\,{y}^{2}-37224\,x+12408\,y-12560.
\end{eqnarray*}
Two observations allow us to obtain a more compact form for $F$. First, looking at the syzygy in the leading term in $x^4$, we see that 
 $$513 \; P(x,y)+216 \; Q(x,y) = 2201580\;({x}^{2}+3\,x-y+1).$$
   Secondly, the discriminant of $P-c\;Q$ shows that  $117\; P+89\; Q$ has a multiple factor,
namely $$ 117\; P(x,y)+89\; Q(x,y) = 755\, \left( 3\,{x}^{2}+6\,x-2\,y+1 \right)  \left( 2+3\,x-y \right)^{2}.$$
It follows that we have the following ``nicer" rational first integral:
	$$\tilde{F}(x,y)={\frac {{x}^{2}+3\,x-y+1}{ \left( 3\,{x}^{2}+6\,x-2\,y+1 \right)   \left( 2+3\,x-y \right) ^{2}}} .$$
This simplification heuristics (using the spectrum) of the expression of a rational first integral to a more compact form can be obtained automatically by running the command \textsf{SimplifyRFI} of our package {\sc RationalFirstIntegrals}.\\

In this example, the generic algorithm \textsf{GenericRationalFirstIntegral} run with \mbox{$N=4$} takes $0.342$ seconds to compute a rational first integral; we see that, though it is $15$ times slower than 
\textsf{HeuristicRationalFirstIntegral} (or \textsf{DeterministicRationalFirstIntegral}), it still has good performances on relatively small degrees.\\


Let us now have a look at the polynomial vector field given by $$A(x,y) = -18\,{x}^{8}{y}^{8}-20\,{x}^{6}{y}^{9}-6\,{x}^{2}{y}^{12}+24\,{x}^{10}
{y}^{3}-6\,{x}^{4}{y}^{9}-4\,{y}^{13}-3\,{x}^{12}-7\,{x}^{2}{y}^{10},
$$ 
$$B(x,y)=2\,x \left( -16\,{x}^{6}{y}^{9}+8\,{x}^{14}-18\,{x}^{4}{y}^{10}-2\,{y}
^{13}+10\,{x}^{8}{y}^{4}-2\,{x}^{2}{y}^{10}-2\,{x}^{10}y-3\,{y}^{11}
 \right),
$$ 
considered by A.~Ferragut in one of his talks concerning \cite{FG}. It admits a rational first integral of degree $18$. We have run our implementations of \textsf{HeuristicRationalFirstIntegral} and \textsf{ProbabilisticRationalFirstIntegral} with the given bounds $N=3,\,6,\,9,\,12,\,15$, and $18$ in the input. The following table presents the outputs and the timings (in seconds) that we have obtained:

\begin{center}
\begin{tabular}{|c||c|c|c|c|c|c|} \hline 
\backslashbox{Algorithm}{$N$} & 3 & 6 & 9 & 12 & 15 & 18  \\ \hline \hline 
Output \textsf{Heuristic}  & ? & ? & ?  & ? & ?  & $F$ \\ \hline 
Time \textsf{Heuristic} & 0.031 & 1.672 & 29.858 & 319.799 & 1735.189 & 19.548 \\ \hline \hline 
Output \textsf{Probabilistic}  & ? & None & None & None & None & $F$ \\ \hline 
Time  \textsf{Probabilistic}  & 0.015 & 0.066 & 1.023 & 5.386 & 28.714 & 252.842 \\ \hline 
\end{tabular}
\end{center}

In the latter table, ? means that our implementation returns ``I don't know" and $F=P/Q$ is the rational first integral of degree $18$ given by $$P(x,y)= -24\,{x}^{2}{y}^{9}+24\,{x}^{10}-24\,{y}^{10}, $$ 
$$\begin{aligned} 
Q(x,y) & = 8\,{x}^{18}-24\,{x}^{12}{y}^{4}+12\,{x}^{14}y+24\,{x}^{6}{y}^{8}-24\,{
x}^{8}{y}^{5}  +  6\,{x}^{10}{y}^{2}-8\,{y}^{12}+44\,{x}^{2}{y}^{9} \\ & -32\,{x
}^{10}-6\,{x}^{4}{y}^{6}+32\,{y}^{10}+{x}^{6}{y}^{3}.
\end{aligned}$$

Note that we obtain approximatively the same timings if we run the deterministic algorithm \textsf{DeterministicRationalFirstIntegral} instead of \textsf{ProbabilisticRationalFirstIntegral}. We can remark that our implementation of \textsf{HeuristicRationalFirstIntegral} is faster in this example than our implementation of \textsf{ProbabilisticRationalFirstIntegral} when there exists a rational first integral whereas \textsf{ProbabilisticRationalFirstIntegral} is much faster at discarding cases when no rational first integral exists. Moreover, we can see that, in this example, \textsf{HeuristicRationalFirstIntegral} only returns ``I don't know" for $N=6, \, 9,\,12,\,15$ whereas in these cases, \textsf{ProbabilisticRationalFirstIntegral} proves that there is no rational first integral of degree at most $N$. Note that these two drawbacks of \textsf{HeuristicRationalFirstIntegral} come from our implementation, which uses the command \textsf{seriestoalgeq} of the {\sc gfun} package, and not from the algorithm itself.

In this example, if we replace $Q$ by $P+\frac{3}{4}\,Q$, we obtain a new rational first integral $\tilde{F}=\frac{P}{P+\frac{3}{4}\,Q}$ which has a ``nicer" (more compact) form
$$\tilde{F}(x,y) = {\frac {{x}^{2}{y}^{9}-{x}^{10}+{y}^{10}}{ \left( 2\,{x}^{6}-2\,{y}^{4
}+{x}^{2}y \right) ^{3}}}.$$

This simplification of the expression of the rational first integral $P/Q$ to a more compact form is obtained with the command \textsf{SimplifyRFI} of our package.


\subsection{A hypergeometric example}
\phantom{$P1 := 4*n^2*(x-1)*(x+1) ;  Q1 := 1+(-4*n^2*x^2+4*n^2)*y^2-4*x*y*n^2 ;$}

Consider the family of polynomial vector fields given by $A=4\,{n}^{2} \left( x-1 \right)  \left( x+1 \right) $ and $B=1+ \left( -4\,{n}^{2}\,{x}^{2}+4\,{n}^{2} \right) {y}^{2}-4\,x\,y\,{n}^{2}$. For each integer $n \in \mathbb{N}^*$, it admits a rational first integral of degree $N=4\,n+1$. This system is derived from the Riccati equation of a standard hypergeometric equation with a finite dihedral differential Galois group, see \cite{HoWe05}. The following table contains the timings (in seconds) for  \textsf{HeuristicRationalFirstIntegral} to find a rational first integral of degree $N=4\,n+1$ when it is run with $N=4\,n+1$.

\begin{center}
 \begin{tabular}{|c||c|c|c|c|c|c|}  \hline  
 $n$ &2&4&6&8&10 \\ \hline 
 Degree $N$ &9&17&25&33& 41\\ \hline \hline 
  Time \textsf{Heuristic} & 0.540 &  12.548 &  118.804 & 592.494 & 3247.325    \\\hline 
  \end{tabular}
  \end{center}
  In short, it takes $2$ minutes to compute a rational first integral of degree $25$ and $52$ minutes to compute a rational first integral of degree $41$ for this family of examples.

\subsection{An Abel equation}
We consider the rationally integrable Abel differential equation $(3)$ in the article of Gine and Llibre  \cite{GiLl10}. It corresponds to the polynomial vector field given by $A(x,y)=x \left( 8\,y-9 \right)$ and $B(x,y)=3\,{y}^{2}-x-3\,y$.  A rational first integral of degree $12$ is computed in $4.142$ seconds by  \textsf{HeuristicRationalFirstIntegral} and in $31.976$ seconds by \textsf{DeterministicRationalFirstIntegral} if they are both run with $N=12$. The rational first integral returned by \textsf{HeuristicRationalFirstIntegral} is given by $P/Q$ with 
 {\small $$\begin{aligned} 
P(x,y) & = 80\,{y}^{12}+480\,x{y}^{10}+1200\,{x}^{2}{y}^{8}-1440\,x{y}^{9}+1600\,
{x}^{3}{y}^{6}-5760\,{x}^{2}{y}^{7}+1200\,{x}^{4}{y}^{4} \\ & -8640\,{x}^{3}
{y}^{5}+8640\,{x}^{2}{y}^{6}+480\,{x}^{5}{y}^{2}-5760\,{x}^{4}{y}^{3}+
13248\,{x}^{3}{y}^{4}+80\,{x}^{6}-1440\,{x}^{5}y \\ & +576\,{x}^{4}{y}^{2}-
13248\,{x}^{3}{y}^{3}-4032\,{x}^{5}+36288\,{x}^{4}y-27216\,{x}^{4}
\end{aligned}$$}
	and
{\small $$\begin{aligned} 
Q(x,y) & =  3\,{y}^{12}+18\,x{y}^{10}+45\,{x}^{2}{y}^{8}-54\,x{y}^{9}+60\,{x}^{3}{
y}^{6}-216\,{x}^{2}{y}^{7}+45\,{x}^{4}{y}^{4}-324\,{x}^{3}{y}^{5} \\ & +324
\,{x}^{2}{y}^{6}+18\,{x}^{5}{y}^{2}-216\,{x}^{4}{y}^{3}+680\,{x}^{3}{y
}^{4}+3\,{x}^{6}-54\,{x}^{5}y+388\,{x}^{4}{y}^{2}-680\,{x}^{3}{y}^{3} \\ & +
32\,{x}^{5}-288\,{x}^{4}y+216\,{x}^{4}
\end{aligned}$$}
Using the \textsf{SimplifyRFI} procedure, we find a rational first integral written in a more compact form: 
$$ F(x,y) = {\frac { \left( {y}^{4}+2\,{y}^{2}x+{x}^{2}-6\,yx \right) ^{3}}{{x}^{3
} \left( 4\,{y}^{4}+8\,{y}^{2}x-4\,{y}^{3}+4\,{x}^{2}-36\,yx+27\,x
 \right) }}$$

\section{Computation of Darboux polynomials}\label{sec:Darboux}
In this section, we show how the approach used above for computing rational first integrals of \eqref{eq-sys} of degree bounded by a fixed $N \in \NN$ can be slightly modified for computing all irreducible Darboux polynomials for the derivation $\DD$ associated with \eqref{eq-sys} of degree at most $N$. \\
In the output of our algorithms, irreducible Darboux polynomials in $\overline{\KK}[x,y]$ will be given by $M(c,x,y) \in \KK[c,x,y]$ and $f(c) \in \KK[c]$. The univariate polynomial $f(c)$ is irreducible in $\KK[c]$ and for all roots $c_i$ of $f(c)$, we have an irreducible Darboux polynomial $M(c_i,x,y) \in \overline{\KK}[x,y]$.\\

\subsection{A deterministic algorithm}
In this section we give a deterministic algorithm for computing all irreducible Darboux polynomials for the derivation $\DD$ associated with \eqref{eq-sys} of degree at most $N$. This algorithm is divided into two steps. First, we compute all irreducible Darboux polynomials $M(x,y)$ such that $M(0,y) \not \in \KK$: this is the task of Algorithm \textsf{IrreducibleDarbouxPolynomialsPartial} below applied to $A$ and $B$. Then, in a second step, we show how we can compute the missing Darboux polynomials (those satisfying $M(0,y) \in \KK$) by applying \textsf{IrreducibleDarbouxPolynomialsPartial} to relevant polynomials constructed from $A$ and $B$ by a change of coordinates. \\

In these algorithms we suppose $A(0,y)\not \equiv 0$ and $A(0,y)$, $B(0,y)$ coprime. We can easily reduce  our study to this situation. We have already explained how we can get $A(0,y)\not \equiv 0$. Now, we just have to remark that the second condition corresponds to the choice of an element which is not a root of the resultant ${\rm Res}_y(A(x,y),B(x,y))$. Thus after a finite number of shifts, we can assume that $A(0,y)\not \equiv 0$ and that $A(0,y)$ and $B(0,y)$ are coprime. In particular, this implies that $x$ is not a Darboux polynomial and if  $M$ is a Darboux polynomial, then $M(0,y) \not \equiv  0$ in $\KK[y]$. We also assume that $\DD$ would have no rational first integral with degree at most $N$. Indeed, from Theorem~\ref{thm:darboux}, in this situation $\DD$ has an infinite number of irreducible Darboux polynomials. We can check this hypothesis with the previous algorithms. \\

\noindent \underline{Algorithm  \textsf{IrreducibleDarbouxPolynomialsPartial}}\\

\noindent \texttt{Input:} $A,\,B\in \KK[x,y]$ s.t. $A(0,y) \not \equiv 0$,  $A(0,y)$, $B(0,y)$ coprime, and a bound $N \in \NN$ such that \eqref{eq-sys} has no rational first integral of degree at most $N$.\\
\noindent \texttt{Output:} The set of all irreducible Darboux polynomials  $M$ for the derivation  $\DD$ such that $\deg(M) \leq N$ and $M(0,y) \not \in \KK$. \\
\begin{enumerate}
\item $\mathcal{E}:=\emptyset$.\\
\item  For an indeterminate $c$, compute the polynomial $y_{c}\in \KK(c)[x]$ of degree at most $ (N^2 +1)$ s.t. $y_{c}(0)=c$ and 
	$\frac{dy_{c}}{dx}\equiv \frac{B(x,y_{c})}{A(x,y_{c})} \mod x^{N^2+1}$.\\
\item  For an indeterminate $c$, compute the polynomial $x_{c}\in \KK(c)[y]$ of degree at most $ (N^2 +1)$ s.t. $x_{c}(c)=0$ and 
	$\frac{dx_{c}}{dy}\equiv \frac{A(x_c,y)}{B(x_c,y)} \mod y^{N^2+1}$.\\
\item Let $M(x,y)=\sum_{i=0}^N \, \left( \sum_{j=0}^{N-i} m_{i,j} \, x^j \right) \, y^i$ be an ansatz for the Darboux polynomials that we are searching for.\\
\item Construct the linear system $\mathcal{L}_1(c)$ for the $m_{i,j}$'s given by:
$$M(x,y_{c}(c,x)) \equiv 0 \mod x^{N^2+1}.$$
\item Construct the linear system $\mathcal{L}_2(c)$ for the $m_{i,j}$'s given by:
$$M(x_c(c,y),y) \equiv 0 \mod y^{N^2+1}.$$
\item For $k=1,\,2$ do:\\
\begin{enumerate}
\item Clear the denominator in $\mathcal{L}_k(c)$.
\item \label{stepdarboux} Compute the Smith normal form of $\mathcal{L}_k(c)$. Let $\mathcal{P}_k(c)$ be the last invariant factor of $\mathcal{L}_k(c)$.
\item Factorize $\mathcal{P}_k(c)$ over $\KK$: $\mathcal{P}_k(c)=\prod_{i=1}^{s_k} \mathcal{P}_{k,i}(c)$.
\item For $i$ from 1 to $s_k$ do:
\begin{enumerate}
\item Set $\KK[c_i]:=\KK[c]/(\mathcal{P}_{k,i}(c))$.
\item \label{darboux_reconnu} Compute a solution of  $\mathcal{L}(c_i)$ s.t. the corresponding polynomial $M_{k,i}$ has minimal degree in $y$ and is primitive w.r.t.~ $y$.
\item If $\gcd(\DD(M_{k,i}),M_{k,i})=M_{k,i}$, then $\mathcal{E}:=\mathcal{E} \cup \{[M_{k,i}(c,x,y),\mathcal{P}_{k,i}(c)]\}$.
\end{enumerate}
\end{enumerate}
\item Return $\mathcal{E}$.\\
\end{enumerate}

\begin{Prop}\label{prop:darbouxpartcorrect}
Algorithm \textsf{IrreducibleDarbouxPolynomialsPartial} is correct.
\end{Prop}

\begin{proof}
Let $M$ be an irreducible Darboux polynomial such that $M(0,y) \not \in \KK$ and $c_M$ be a root of $M(0,y)$.
Then we have: $A(0,c_M)\neq 0$ or $B(0,c_M)\neq 0$ because $A(0,y)$ and $B(0,y)$ are assumed to be coprime.\\
If $A(0,c_M)\neq 0$ and $M(0,c_M)= 0$, then $M$ admits a root $y_{c_M} \in  \overline{\KK(x)}$ such that \mbox{$y_{c_M}(0)=c_M$}.
Then, from Proposition \ref{RFI-Darboux-Diffeq}, $y_{c_M}$ is a power series solution of \eqref{diff-eq}.
 Thus $c_M $ is a root of $\mathcal{P}_1(c)$.    Then, by Lemma \ref{algebraic-solution}, $M$ is constructed in Step~(\ref{darboux_reconnu}).\\
  If for a constant $c_M$, we have $B(0,c_M)\neq 0$ and $M(0,c_M)= 0$, then the previous arguments used with  $\mathcal{P}_2(c)$ show that $M$ is also constructed.
\end{proof}

In the algorithm \textsf{IrreducibleDarbouxPolynomialsPartial}, we compute irreducible Darboux polynomials $M$ such that $M(0,y) \not \in \KK $. Indeed, the algorithm finds a irreducible Darboux polynomial $M$ if and only if the curve $M(x,y)=0$ and the line $x=0$ have an intersection point. Now, we show how to get irreducible Darboux polynomials such that $M(0,y) \in \KK$. The idea is to use a change of coordinates in order to get a new polynomial $\tilde{M}$ such that $\tilde{M}(0,y)$ has a root. If $M(0,y) \in \KK$, then $M$ has a root at infinity. Thus we consider the following change of coordinates: we set 
$$A^{\sharp}(x,y,z)=A\Big(\frac{x}{z},\frac{y}{z}\Big)\,z^d, \,\, B^{\sharp}(x,y,z)=B\Big(\frac{x}{z},\frac{y}{z}\Big)\,z^d, \, \,M^{\sharp}(x,y,z)=M\Big(\frac{x}{z},\frac{y}{z}\Big)\,z^k,$$
 where $k=\deg(M)$, and we consider the following polynomials:
$$\tilde{A}(y,z)=A^{\sharp}(1,y,z),  \,\,  \tilde{B}(y,z)=B^{\sharp}(1,y,z),   \,\, \tilde{M}(y,z)=M^{\sharp}(1,y,z).$$

A straightforward computation shows that:
\begin{Lem}\label{lem:chgcoord}
With the above notation, if $M$ is a Darboux polynomial for the derivation $\DD=A(x,y)\,\frac{\partial}{\partial x} + B(x,y)\,\frac{\partial}{\partial
y}$, then $\tilde{M}$ is a Darboux polynomial for the derivation  
$$\tilde{\DD}=\left(-y\,\tilde{A}(y,z)+\tilde{B}(y,z) \right)\, \frac{\partial}{\partial y} -\tilde{A}(y,z)\,z \, \frac{\partial}{\partial z}.$$
Furthermore, if $M(0,y) \in \KK\setminus\{0\}$, then $\tilde{M}(0,z) \not \in \KK $.
\end{Lem}

We deduce the following algorithm:\\

\noindent \underline{Algorithm  \textsf{IrreducibleDarbouxPolynomials}}\\

\noindent \texttt{Input:} $A,\,B\in \KK[x,y]$ s.t. $A(0,y) \not \equiv 0$,  $A(0,y)$, $B(0,y)$ coprime,  $\tilde{B}(0,z)\not \equiv 0$, $\tilde{A}(0,y)$ and $\tilde{B}(0,y)$ coprime, and a bound $N \in \NN$ such that \eqref{eq-sys} has no rational first integral of degree at most $N$.\\
\noindent \texttt{Output:} The set of all irreducible Darboux polynomials  $M$ for the derivation  $\DD$ such that $\deg(M) \leq N$. \\
\begin{enumerate}
\item $\mathcal{E}:=\textsf{IrreducibleDarbouxPolynomialsPartial}(A,B,N)$.
\item $\mathcal{E}':=\textsf{IrreducibleDarbouxPolynomialsPartial}(-y\,\tilde{A}+\tilde{B}, -\tilde{A}\,z,N)$.
\item For all $[\tilde{M}(c,y,z),\mathcal{P}(c)] \in \mathcal{E}'$ do:
\begin{enumerate}
\item  $M(c,x,y):=\tilde{M}(c,\frac{y}{x},\frac{1}{x})\,x^{\deg(M)}$.
\item Add $[M(c,x,y),\mathcal{P}(c)]$ to $\mathcal{E}$.
\end{enumerate}
\item Return $\mathcal{E}$.\\
\end{enumerate}

For the same reasons as before, using a finite number of shifts we can suppose that the hypotheses ``$A(0,y)\not \equiv 0$, $A(0,y)$, $B(0,y)$ coprime,  $\tilde{B}(0,z)\not \equiv 0$, $\tilde{A}(0,y)$, $\tilde{B}(0,y)$ coprime" are satisfied so that these conditions are not restrictive.\\

As a direct consequence of Proposition~\ref{prop:darbouxpartcorrect} and Lemma~\ref{lem:chgcoord}, we obtain the following result.

\begin{Prop}
Algorithm \textsf{IrreducibleDarbouxPolynomials} is correct.
\end{Prop}

\subsection{A probabilistic algorithm}
As we have seen in Section \ref{complexity-section},  the computation of a basis of solutions of a system of linear equations is the most costly step of our algorithms. In \textsf{IrreducibleDarbouxPolynomials}, we have to consider four systems of linear equations. The first reason is that in \textsf{IrreducibleDarbouxPolynomialsPartial}, we need to study two linear systems in order to take into account the situation where $x=0$ is a vertical tangent of the curve $M(x,y)=0$. Indeed, in this situation we can not get a parametrization $\big(x,y(x)\big)$ of the curve. The second reason is that we need to use a change of coordinates in order to control the situation where $M(0,y)$ has a root at infinity. Of course, for a generic polynomial vector field, these two situations (i.e., a vertical tangent and a root at infinity) do not appear. We then deduce the following probabilistic algorithm.\\

\noindent \underline{Algorithm  \textsf{ProbabilisticIrreducibleDarbouxPolynomials}}\\

\noindent \texttt{Input:}  $A,\,B\in \KK[x,y]$, a bound $N \in \NN$ such that \eqref{eq-sys} has no rational first integral of degree at most $N$, and two elements $x_0, \, \alpha \in \KK$.\\
\noindent \texttt{Output:} The set of all irreducible Darboux polynomials  $M$ for the derivation  $\DD$ such that of $\deg(M) \leq N$. \\ \begin{enumerate}
\item $\mathcal{E}:=\emptyset$.
\item Set $A_{\alpha}(x,y)=A(x+\alpha \, y,y) -\alpha \, B(x+\alpha \, y,y)$,  $B_{\alpha}(x,y)=B(x+\alpha \, y,y)$ and $\DD_{\alpha}=A_{\alpha}(x,y) \,\frac{\partial}{\partial x} +B_{\alpha}(x,y)\,\frac{\partial}{\partial y}$.
\item  For an indeterminate $c$, compute the polynomial $y_{c}\in \KK(c)[x]$ of degree $\leq (N^2 +1)$ s.t. $y_{c}(x_0)=c$ and 
	$\frac{dy_{c}}{dx}\equiv \frac{B_{\alpha}(x,y_{c})}{A_{\alpha}(x,y_{c})} \mod x^{N^2+1}$.
\item Let $M(x,y)=\sum_{i=0}^N \, \left( \sum_{j=0}^{N-i} m_{i,j} \, x^j \right) \, y^i$ be an ansatz for the Darboux polynomials that we are searching for.
\item Construct the linear system $\mathcal{L}(c)$ for the $m_{i,j}$'s given by:
$$M(x,y_{c}(c,x)) \equiv 0 \mod x^{N^2+1}.$$
\item Clear the denominator in $\mathcal{L}(c)$.
\item \label{stepdarboux} Compute the Smith normal form of $\mathcal{L}(c)$. Let $\mathcal{P}(c)$ be the last invariant factor of $\mathcal{L}(c)$.
\item Factorize $\mathcal{P}(c)$ over $\KK$: $\mathcal{P}(c)=\prod_{i=1}^{s} \mathcal{P}_{i}(c)$.
\begin{enumerate}
\item For i from 1 to $s$ do:
\begin{enumerate}
\item Set $\KK[c_i]=\KK[c]/(\mathcal{P}_i(c))$.
\item \label{darboux_reconnu_1} Compute a solution of  $\mathcal{L}(c_i)$ s.t. the corresponding polynomial $M_i$ has minimal degree in $y$ and is primitive w.r.t.~ $y$.
\item If $\gcd(\DD_{\alpha}(M_i),M_i)=M_i$, then $\mathcal{E}:=\mathcal{E} \cup \{[M_i(c,x-\alpha \,y,y),\mathcal{P}_i(c)]\}$.
\end{enumerate}
\end{enumerate}
\item Factorize $A_{\alpha}(x,0)$ over $\KK$: $A_{\alpha}(x,0)=\prod_{i=1}^k A_i(x)$.
\item For i from 1 to $k$ do: 
\begin{enumerate}
\item If $\gcd(\DD_{\alpha}(A_i),A_i)=A_i$, then $\mathcal{E}:=\mathcal{E} \cup \{[A_i(x-\alpha \, y,y),c-1]\}$.
\end{enumerate}
\item Return $\mathcal{E}$.\\
\end{enumerate}

\begin{Prop}
The algorithm  \textsf{ProbabilisticIrreducibleDarbouxPolynomials} is correct. Furthermore, if $x_0$ and $\alpha$ are chosen uniformly at random in a finite set $\Omega \subset \KK$ such that $|\Omega| > Nd\,(\mathcal{B}(d)+1)$, then the probability that this algorithm returns \underline{\emph{all}}  irreducible Darboux polynomials  is at least $\Big( 1 - \dfrac{N\,(\mathcal{B}(d)+1)}{|\Omega|}\Big) \, \Big( 1 - \dfrac{N\,d\,(\mathcal{B}(d)+1)}{|\Omega|}\Big) $.
\end{Prop}

\begin{proof}
First, we remark that $M$ is a Darboux polynomial for $\DD$ if and only if $M_{\alpha}(x,y)=M(x+\alpha y,y)$ is a Darboux polynomial for $\DD_{\alpha}$. Thus the strategy used in this algorithm is to perform a change of coordinates in order to be in a generic position, and then to compute all irreducible Darboux polynomials $M_{\alpha}$ of degree at most $N$ by considering only one linear system.\\
The proof  of Proposition \ref{prop:darbouxpartcorrect} shows that from Step (2) to Step (8), we compute all irreducible Darboux polynomials satisfying: 
$$M_{\alpha}(x_0,y) \not \in \KK  \quad \textrm{ and } \quad   {\rm Res}_y(M_{\alpha}(x_0,y),A_{\alpha}(x_0,y))\neq 0.$$
Let us study the probability to get $M_{\alpha}(x_0,y) \not \in \KK$. If $M(x,y)=\sum_{0\leq i+j\leq N}a_{i,j}\,x^i\,y^j$, then $M_{\alpha}(x_0,y)=(\sum_{i+j=N}a_{i,j}\,\alpha^i)\,y^N+ \cdots$ where the other terms have degree relatively to $y$ strictly less than $N$. Thus, if $\sum_{i+j=N}a_{i,j}\alpha^i$ is not equal to zero, then we have  $M_{\alpha}(x_0,y) \not \in \KK$.\\
 As \eqref{eq-sys} has no rational first integral of degree at most $N$, then by Darboux-Jouanolou's theorem (see Theorem \ref{thm:darboux}), we have at most $\mathcal{B}(d)+1$ irreducible Darboux polynomials with degree at most $N$. Thus, by  Zippel-Schwartz's lemma, the probability to reach the situation $M_{\alpha}(x_0,y) \not \in \KK$ for all irreducible Darboux polynomials is at least $1-(\mathcal{B}(d)+1)\,N/|\Omega|$.\\
Now we suppose that  $M_{\alpha}(x_0,y) \not \in \KK$ and we study the probability to have the situation ${\rm Res}_y(M(x_0,y),A(x_0,y))\neq 0$.
If the polynomial  ${\rm Res}_y(M(x,y),A(x,y))$ is not zero, then, by  Zippel-Schwartz's lemma, the probability to reach this situation for all irreducible Darboux polynomials, is at least $1-(\mathcal{B}(d)+1)\,N\,d/|\Omega|$.\\
If the polynomial  ${\rm Res}_y(M_{\alpha}(x,y),A_{\alpha}(x,y))$ is zero, then $M_{\alpha}$ and $A_{\alpha}$ have a common  factor. As we  suppose  $M_{\alpha}$  irreducible, we deduce that $M_{\alpha}$ divides $A_{\alpha}$. Thus $M_{\alpha}$ divides $B_{\alpha}\,\partial_y (M_{\alpha})$. As $A_{\alpha}$ and $B_{\alpha}$ are coprime, we get that $M_{\alpha}$ divides $\partial_y(M_{\alpha})$. This situation is possible only when $\deg_y(M_{\alpha})=0$. This means  ${\rm Res}_y(M_{\alpha}(x,y),A_{\alpha}(x,y)) \equiv 0$ when $\deg_y(M_{\alpha})=0$ and  $M_{\alpha}$ divides $A_{\alpha}(x,0)$. We compute this kind of irreducible Darboux polynomials in Step (10) of the algorithm.\\
In conclusion, the algorithm computes all irreducible Darboux polynomials of degree at most $N$ with the announced probability estimate.\end{proof}
\subsection{Implementation and example}
We have implemented the algorithm \textsf{ProbabilisticIrreducibleDarbouxPolynomials} in our package {\sc RationalFirstIntegrals}\footnote{It is available at \url{http://www.ensil.unilim.fr/~cluzeau/RationalFirstIntegrals.html}}. Let us illustrate the purpose of this section on an interesting example. \\
Consider the vector field corresponding to the jacobian derivation associated with $f(x,y)=(y-x-1)\,(x-y^2)\,(x\,y-1)$, namely,  
$$A(x,y):=-\frac{\partial f}{\partial y}(x,y)=-3\,{x}^{2}{y}^{2}+4\,x{y}^{3}+{x}^{3}-2\,{x}^{2}y-3\,x{y}^{2}+{x}^{2}+2\,xy-3\,{y}^{2}+x+2\,y,$$
$$B(x,y):=\frac{\partial f}{\partial x}(x,y)=2\,x{y}^{3}-{y}^{4}-3\,{x}^{2}y+2\,x{y}^{2}+{y}^{3}-2\,xy-{y}^{2}+2\,x-y+1.$$
By construction, it admits the rational first integral $f$ of degree $4$ and the Darboux polynomials $M_1(x,y)=y-x-1$, $M_2(x,y)=x-y^2$, 
$M_3(x,y)=x\,y-1$ of degree at most $2$.  Let us consider the computation of all irreducible Darboux polynomials of degree at most $N=2$.\\
The first Darboux polynomial $M_1$ satisfies $M_1(0,y)=y-1 \not \in \KK$ and its root $c_{M_1}=1$ satisfies $A(0,c_{M_1})=-1 \neq 0$. Therefore it will be found by considering the linear system $\mathcal{L}_1(c)$ in \textsf{IrreducibleDarbouxPolynomialsPartial}, see the proof of Proposition~\ref{prop:darbouxpartcorrect}. \\
The Darboux polynomial $M_2$ satisfies $M_2(0,y)=-y^2 \not \in \KK$ and but its root $c_{M_2}=0$ satisfies $A(0,c_{M_2})=0$. Thus it will be missed if we only consider system $\mathcal{L}_1(c)$ in \textsf{IrreducibleDarbouxPolynomialsPartial}. It is the case where the curve $M_2(x,y)=0$ has the vertical tangent $x=0$. However, if we consider the second system $\mathcal{L}_2(c)$ in \textsf{IrreducibleDarbouxPolynomialsPartial}, we will find this Darboux polynomial, see the proof of Proposition~\ref{prop:darbouxpartcorrect}. \\
Finally  $M_3$ satisfies $M_3(0,y)=-1 \in \KK$ so that $M_3(0,y)$ has a root at infinity. Considering only the systems  $\mathcal{L}_1(c)$ and $\mathcal{L}_2(c)$ in \textsf{IrreducibleDarbouxPolynomialsPartial} will not be enough to find this Darboux polynomial. However, performing the change of coordinates as in \textsf{IrreducibleDarbouxPolynomials} and applying \textsf{IrreducibleDarbouxPolynomialsPartial} to  $-y\,\tilde{A}+\tilde{B}$ and  $-\tilde{A}\,z$ instead of $A$ and $B$ will provide this Darboux polynomial.\\
To summarize, applying \textsf{IrreducibleDarbouxPolynomialsPartial} to $A$ and $B$, we get $M_1$ and $M_2$ but we miss $M_3$ but either applying  \textsf{IrreducibleDarbouxPolynomials} or  \textsf{ProbabilisticIrreducibleDarbouxPolynomials} we get the three Darboux polynomials. Note also that applying an algorithm similar to  \textsf{ProbabilisticIrreducibleDarbouxPolynomials} but where we skip Step (2), i.e., we do not perform the generic change of coordinate, we would obtain only $M_1$ and miss both $M_2$ and $M_3$. \\
Our implementation of  \textsf{ProbabilisticIrreducibleDarbouxPolynomials} requires computations in $\KK(c)$ so that as for \textsf{GenericRationalFirstIntegral} it is not very efficient and can not be used in practice for examples with large degrees. To give an idea of timings\footnote{All the computations were made on a 2.7 GHz  Intel Core i7}, on the previous example, running 
 \textsf{ProbabilisticIrreducibleDarbouxPolynomials} without the change of coordinates in Step (2), we obtain $\{M_1\}$ in $2.737$ seconds but we miss $M_2$ and $M_3$ whereas running \textsf{ProbabilisticIrreducibleDarbouxPolynomials}, we get the complete set $\{M_1,M_2,M_3\}$ in $9.775$ seconds.



\newcommand{\etalchar}[1]{$^{#1}$}

\end{document}